%% file: handling-algebraic-effects.tex
\documentclass{LMCS}

\def\dOi{9(4:23)2013}
\lmcsheading%
{\dOi}
{1--36}
{}
{}
{Oct.~20, 2011}
{Dec.~17, 2013}
{}


\ACMCCS{[{\bf Theory of computation}]: Semantics and
  reasoning---Program constructs; Semantics and reasoning---Program
  semantics---Algebraic semantics}

\usepackage{amsmath, amssymb}
\usepackage{bbm}
\usepackage[nohug]{diagrams}
\usepackage{mathpartir}
\usepackage[only, llbracket, rrbracket, rightarrowtriangle]{stmaryrd}
\input{definitions}

\usepackage[
  pdfauthor={Gordon {D.} Plotkin and Matija Pretnar},
  pdftitle={Handling Algebraic Effects}
]{hyperref}

\begin{document}

\title[Handling Algebraic Effects]{Handling Algebraic Effects\rsuper*}

\author[G.~D.~Plotkin]{Gordon D.\ Plotkin\rsuper a}
\address{{\lsuper a}Laboratory for Foundations of Computer Science, School of Informatics, University of Edinburgh, Scotland}
\email{gdp@inf.ed.ac.uk}
\thanks{{\lsuper a}This research was supported by EPSRC grant GR/586371/01 and by a Royal Society-Wolfson Award.}

\author[M.~Pretnar]{Matija Pretnar\rsuper b}
\address{{\lsuper b}Faculty of Mathematics and Physics, University of Ljubljana, Slovenia}
\email{matija.pretnar@fmf.uni-lj.si}

\titlecomment{{\lsuper*}A preliminary version of this work was presented at ESOP 2009, see~\cite{plotkin09handlers}.}

\begin{abstract} \noindent
  Algebraic effects are computational effects that can be represented by an equational theory whose operations produce the effects at hand.
  The free model of this theory induces the expected computational monad for the corresponding effect.
  Algebraic effects include exceptions, state, nondeterminism, interactive input/output, and time, and their combinations.
  Exception handling, however, has so far received no algebraic treatment.

  We present such a treatment, in which
    each handler yields a model of the theory for exceptions, and
    each handling construct yields the homomorphism induced by the universal property of the free model.
  We further generalise exception handlers to arbitrary algebraic effects.
  The resulting programming construct includes many previously unrelated examples from both theory and practice,
    including relabelling and restriction in Milner's CCS, timeout, rollback, and stream redirection.
\end{abstract}

\maketitle

\section*{Introduction}

In seminal work~\cite{moggi91notions}, Moggi proposed a uniform representation of computational effects by monads~\cite{benton00monads}.
For example, working in the category of sets, a computation that returns values from a set $A$ is modelled by an element of $T A$ for a suitable monad $T$.
Examples of such effects include exceptions, state, nondeterminism, interactive input/output, time, continuations, and combinations of them.
Later, Plotkin and Power proposed to represent effects by
\begin{enumerate}
  \item a set of operations that represent the sources of effects; and
  \item an equational theory for these operations 
  that describes their properties~\cite{plotkin04computational}.
\end{enumerate}

The basic operational intuition is that each computation either returns a value or performs an operation with an outcome that determines a continuation of the computation.
The arguments of the operation represent the possible continuations.
For example, using a binary choice operation $\chooseop*$, a computation that nondeterministically chooses a boolean is:
\[
  \chooseop{}(\ret \tru, \ret \fls)
\]
The outcome of making the choice is binary:
  either to continue with the computation given by the first argument
  or else to continue with that given by the second.

A computation that returns values from a set $A$ is modelled by an element of the free model $F A$, generated by the equational theory.
In the case of nondeterminism, these equations state that $\chooseop*$ is a semilattice operation.
Modulo the forgetful functor, the free model functor is exactly the monad proposed by Moggi to model the corresponding effect~\cite{plotkin02notions}.
Effects whose monad can be obtained by such an equational presentation are called \emph{algebraic}; the relevant monads are exactly the ranked ones.
With the notable exception of continuations~\cite{flanagan93the-essence,hyland07combining},
  all of the above effects are algebraic, and, indeed, have natural equational presentations.

The algebraic view has given ways of combining effects~\cite{hyland06combining} and reasoning about them~\cite{plotkin08a-logic}.
However, exception handlers provided a challenge to the algebraic approach.

The monad $\mathord{-} + \exc$ for a given set of exceptions $\exc$ is presented by a nullary exception raising operation $\raiseop{e}()$ for each $e \in \exc$ and no equations.
The operation $\raiseop{e}()$ takes no arguments as there is no continuation immediately after an exception has been raised.

The question then arises how to deal with exception handling.
One approach would be to consider a binary exception handling construct
\[
  \handleop{e}(M, N)
\]
which proceeds as $M$ unless the exception $e$ is raised, when it proceeds as $N$.
This construct has a standard interpretation as a binary operation using the exception monad;
  however, as explained in~\cite{plotkin03algebraic}, it lacks a certain naturality property characterising equationally specified operations.
In programming terms this corresponds to the operation commuting with evaluation contexts $\mathcal{E}[-]$.
For example we would expect each of the following two equations to hold:
\[
  \mathcal{E}[\chooseop{}(M, N)] = \chooseop{}(\mathcal{E}[M], \mathcal{E}[N])
  \qquad
  \mathcal{E}[\raiseop{e}()] = \raiseop{e}()
\]
but not:
\[
  \mathcal{E}[\handleop{e}(M, N)] = \handleop{e}(\mathcal{E}[M],\mathcal{E}[N])
\]
Since the naturality property is common to all equationally specified operations, it follows that no alternative (ranked) monad will suffice either.

In this paper we give an algebraic account of exception handling.
The main idea is that
\begin{enumerate}
  \item handlers correspond to (not necessarily free) models of the equational theory; and
  \item the semantics of handling is given using unique homomorphisms that target such models and are induced by the universal property of the free model.
\end{enumerate}
The usual exception handling construct corresponds to the application of the unique homomorphism that preserves returned values.
We, however, adopt a more general approach, suggested by Benton and Kennedy~\cite{benton01exceptional}, and which was an inspiration for the present work.
In Benton and Kennedy's approach returned values are passed to a user-defined continuation;
  this amounts to an application of an arbitrary induced homomorphism to a computation.

As we shall see, this idea generalises to all algebraic effects, yielding a new programming concept enabling one to handle any algebraic effect.
Examples include relabelling and restriction in CCS~\cite{milner89-calculus}, timeout, rollback, stream redirection of shell processes, and many others.
Conceptually, algebraic operations and effect handlers are dual:
  the former could be called \emph{effect constructors} as they give rise to the effects;
  the latter could be called \emph{effect deconstructors} as the computations they provide proceed according to the effects already created.
Filinski's reflection and reification operations provide general effect constructors and deconstructors in the context of layered monads~\cite{filinski99representing}.

In Section~\ref{sec:exception_handlers}, we illustrate the main semantic ideas via an informal discussion of exception handlers.
Then, in Section~\ref{sec:syntax}, we give a calculus extending Levy's call-by-push-value~\cite{levy06call-by-push-value} with operations, handler definitions, and an effect handling construct, which handles computations using a given handler.
In Section~\ref{sec:examples}, we give some examples that demonstrate the versatility of our handlers.
Next, in Section~\ref{sec:semantics}, we provide a denotational semantics, and define a notion of handler correctness;  an informal introduction to handler correctness is given in Section~\ref{sub:handling_arbitrary_effects}.
The denotational semantics is given in terms of sets and functions although a more general categorical semantics should also be possible.

In Section~\ref{sec:reasoning_about_handlers}, we sketch some reasoning principles for handlers, and then, in Section~\ref{sec:correctness_of_handlers}, we give some results on the difficulty of deciding handler correctness.
In Section~\ref{sec:recursion}, we describe the inclusion of recursion; to do this we switch from the category of sets and functions to that of \emph{$\omega$-cpos} (partial orders with suprema of increasing countable chains) and \emph{continuous functions} (monotone functions preserving such suprema).
In the conclusion, we list open questions and briefly discuss some possible answers.
At various points in the paper we use operational ideas to aid understanding; we do not however present a formal operational semantics of effect handlers.

\section{Exception Handlers}
\label{sec:exception_handlers}

We start our study with exception handlers,
  both because they are an established concept~\cite{benton01exceptional, levy06monads}
  and also since exceptions are the simplest example of an algebraic effect.
To focus on the exposition of ideas, we write this section in a rather informal style, mixing syntax and semantics.

We consider a finite set of exceptions $\exc$.
Computations returning values from a set $A$ are modelled by elements of the exception monad $T A \defeq A + \exc$.
This has unit
  $\eta_A \defeq \In{1} \from A \to A + \exc$,
and the computation $\ret V$ is interpreted by
  $\eta_A(V) = \In{1}(V)$,
while $\raiseop{e}()$ is interpreted by $\In{2}(e)$.

\subsection{Simple handling construct}

\newcommand{\exchandler}{\{ \raiseop{e}() \maps M_e \}_{e \in \exc}}

Fixing $A$, the simple, standard, handling construct is
\[
  \handle{M} \exchandler
\]
where $\{ \mathord{\cdots} \}_{e \in \exc}$ represents a set of computations, one for each exception $e \in \exc$.
The construct proceeds by carrying out the computation $M \in A + \exc$,
  intercepting raised exceptions $e \in \exc$ by carrying out predefined computations $M_e \in A + \exc$ instead.
If we choose not to handle a particular exception $e$, we take $M_e$ to be $\raiseop{e}()$.
The handling construct satisfies two equations:
\begin{align*}
  \handle{\ret V} \exchandler &= \In{1}(V) \\
  \handle{\raiseop{e'}()} \exchandler &= M_{e'}
\end{align*}

From an algebraic point of view, the computations $M_e$ give a new model $\model$ for the theory of exceptions.
The carrier of this model is $A + \exc$ as before; however, for each $e$, $\raiseop{e}()$ is instead interpreted by $M_e$.
We then see from the above two equations that
\[
  \homo(M) \defeq \handle{M} \exchandler
\]
is the unique homomorphism (a map preserving operations) from $A + \exc$ to $\model$ that extends the map
$\In{1} \from A \to A + \exc$, i.e., so 
that the following diagram commutes:
\begin{diagram}
  A             &                 & \\
  \dTo<{\eta_A} & \rdTo^{\In{1}}  & \\
  A + \exc      & \rDotsto_\homo  & \model
\end{diagram}

The idea is therefore to obtain such a homomorphism~$h$ using the freeness of the model $A + \exc$, used for computations, and then interpret the handling construct as an application of $h$ to the computation being handled.
The existence of $h$, in turn, requires a model $\model$ on $A + \exc$ for its target;
  that model is supplied via the handling construct.

\subsection{Extended handling construct}

\renewcommand{\exchandler}{\{ \raiseop{e}() \maps N_e \}_{e \in \exc}}

Benton and Kennedy~\cite{benton01exceptional} generalised the handling construct to one of the form
\[
  \handleto{M}{\exchandler}{x \T A} N(x)
\]
(written using our syntax).
Here returned values are passed to a user-defined continuation,
  a map $N \from A \to B + \exc$, where $B$ is a set that may differ from $A$; 
  handling computations $N_e$ return values in $B$, if they do not themselves raise exceptions: thus $N_e \in B + \exc$.
These two facts can be expressed equationally:
\begin{align*}
  \handleto{\ret V}{\exchandler}{x \T A} N(x) &= N(V) \\
  \handleto{\raiseop{e'}()}{\exchandler}{x \T A} N(x) &= N_{e'}
\end{align*}
As discussed in~\cite{benton01exceptional}, this construct
  captures a programming idiom that was cumbersome to write with the simpler construct,
  allows additional program optimisations, and
  has a stack-free small-step operational semantics.

Algebraically we again have a model $\model$, this time on $B + \exc$, interpreting $\raiseop{e}()$ by $N_e$.
The handling construct can be interpreted as $h(M)$,
  where $\homo \from A + \exc \to \model$ is the unique homomorphism that extends $N$,
  i.e., such that the following diagram commutes:
\begin{diagram}
  A             &                 & \\
  \dTo<{\eta_A} & \rdTo^{N}       & \\
  A + \exc      & \rDotsto_\homo  & \model
\end{diagram}
Note that \emph{all} the homomorphisms from the free model to a model on a  given carrier are obtained in this way.
So Benton and Kennedy's handling construct is the most general one possible from the algebraic point of view.

\subsection{Handling arbitrary algebraic effects}
\label{sub:handling_arbitrary_effects}

We can now see how to give handlers for other algebraic effects.
A model of an equational theory is an interpretation, i.e., a set and a set of maps, one for each operation, that satisfies the equations;
handlers give definitions of such interpretations.
As before, computations are interpreted in the free model and handling constructs are interpreted by the induced homomorphisms.
Where exceptions were replaced by handling computations,
  operations are now replaced by the handling maps;
  as computations are built from combinations of operations,
  handling a computation may involve several such replacements.

Importantly however, not all interpretations yield a model of the theory, and so, in that sense, not all handlers need be correct (see Section~\ref{sec:semantics}). If a handler is not correct then handling constructs using it have no meaning.
Even more, whether or not a handler is correct  may be undecidable  (see Section~\ref{sec:correctness_of_handlers}).

We can see two general approaches to this difficulty when designing programming languages with facilities to handle effects.
One is to make the language designer responsible:
  the definable families should be restricted so that only models can be defined; this is the approach taken in~\cite{plotkin09handlers}.
Another, the one adopted in this paper, is to allow complete freedom in the language: all possible definitions are permitted.
In this case not all handlers are correct, 
  and it is not the responsibility of the language designer to ensure that all handlers defined are correct.

One can envisage responsibility being assigned variously
  to the language designer,
  to the compiler,
  to the programmer,
  to a protocol for establishing program correctness,
  or, in varying degrees, to all of them.
For example, some handlers may be ``built-in",  and so the responsibility of the language designer,
  while others may be defined by the programmer, and so their responsibility.

\section{Syntax}
\label{sec:syntax}


\subsection{Signatures}

\emph{Signature types} $\alpha, \beta$ are given by:
\[
  \alpha, \beta \bnfis
    \mathbf{b} \bnfor
    \unit \bnfor
    \alpha \times \beta \bnfor
    \sumtype{\ell \in L} \alpha_\ell
\]
  where $\mathbf{b}$ ranges over a given set of \emph{base types},
  where $L$ ranges over finite subsets of \emph{labels} (taken from a fixed set $\mathrm{Lab}$ of all possible labels);
  we do not specify the set of labels in detail, but assume such ones as needed are available.
We specify a subset of the base types as \emph{arity base types},
  and say that a signature type is an \emph{arity signature type} if the only base types it contains are arity ones.

To represent finite data such as
  booleans~$\bool = \{ \tru, \fls \}$,
  finite subsets of integers~$\type{n} = \{ 1, 2, \dots, n \}$,
  the empty set~$\void$,
  characters~$\chr$,
  or memory locations~$\loc$,
we may confuse a finite set $L$ with the type $\sumtype{\ell \in L} \unit$.
For infinite sets, such as natural numbers~$\nat$ or strings, we take base types as needed.

Next, we assume given a set of typed
  \emph{function symbols}~$\f \T \alpha \to \beta$.
These represent pure built-in functions, for example
  arithmetic function symbols, such as $+ \T \nat \times \nat \to \nat$, 
 or   function symbols for arithmetic relations, such as $< \T \nat \times \nat \to \bool$.

Finally, we assume given a finite set of typed
  \emph{operation symbols}~$\op* \T \optype{\alpha}{\beta}$,
    where each such $\beta$ is an arity signature type,
    and an \emph{effect theory}~$\thy$.
The operations represent sources of effects and the effect  theory determines their properties.
In order to to focus on handlers, we postpone the consideration of effect theories to Section~\ref{sub:effect_theory}.

The typing $\op* \T \optype{\alpha}{\beta}$ indicates that the operation represented by $\op*$ accepts a parameter of type $\alpha$ and,
  after performing the relevant effect, its outcome, of type $\beta$, determines its continuation;
we say that $\op*$ is \emph{parameterised} on $\alpha$ and has \emph{arity} $\beta$, or is $\beta$-ary.
In case $\alpha = \unit$, we may just write $\op* \T {\beta}$.

The given sets of base types, 
arity base types, 
typed function symbols, and typed operation symbols constitute a \emph{signature};
the syntax is parameterised by the choice of such a signature.
The choice of signature (and theory) that we take evidently depends on the effects we want to represent.

We now give some examples, taken from~\cite{hyland06combining}.
\begin{exas}
  \label{exa:effects} \hfill
  \begin{description}
    \item[Exceptions]
      We take a single nullary (i.e., $\void$-ary) operation symbol $\raiseop* \T \optype{\exc}{\void}$, parameterised on $\exc$, for raising exceptions.
      Here $\exc$ is a finite set of exceptions; we could instead take it to be a base type if we wanted an infinite set of exceptions.
      The operation symbol is nullary as there is no continuation after raising an exception: instead the exception has to be handled.
    \item[State]
      We take an arity base type $\nat$ for natural numbers and  read and write operation symbols
        $\getop* \T \optype{\loc}{\nat}$ and $\setop* \T \optype{\loc \times \nat}{\unit}$,
      where $\loc$ is a finite set of \emph{locations}.
      The idea is that there is a state holding natural numbers in the locations, and $\getop*$ retrieves the number from a given location,
      while $\setop*$ sets a given location to a given number and returns nothing.
    \item[Read-only state]
      Here we  only take an arity base type $\nat$ and a read operation symbol $\getop* \T \optype{\loc}{\nat}$.
   \item[(Binary) nondeterminism]
      We take a binary (i.e., $\type{2}$-ary) operation symbol $\chooseop*$ for nondeterministic choice.
      The operation symbol is binary as the outcome of making a choice is to decide between one of two choices.
    \item[Interactive input and output (I/O)]
      We take a finite set $\chr$ of characters, and operation symbols
        $\readop* \T {\chr}$, for reading characters, and
        $\writeop* \T \optype{\chr}{\unit}$, for writing them.
      The operation symbol $\readop*$ is $\chr$-ary, as the outcome of reading is to obtain a character;
      the operation symbol $\writeop*$ is unary as there will be just one continuation after writing a character.
  \end{description}
\end{exas}

\subsection{Types}

Our language follows Levy's call-by-push-value approach~\cite{levy06call-by-push-value} and so has a strict separation between \emph{value types}~$A$, $B$ and \emph{computation types}~$\C$.
These types are given by:
\begin{align*}
  A, B &\bnfis
    \mathbf{b} \bnfor
    \unit \bnfor
    A \times B \bnfor
    \sumtype{\ell \in L} A_\ell \bnfor
    U \C \\
  \C &\bnfis
    F A \bnfor
    \prodtype{\ell \in L} \C_\ell \bnfor
    A \to \C
\end{align*}
The value types extend the signature types, as there is an additional type constructor $U \mathord{-}$.
The type $U \C$ classifies computations of type $\C$ that have been \emph{thunked} (or frozen) into values;
such computations can be passed around and later \emph{forced} back into evaluation.

The computation type $F A$ classifies the computations that return values of type $A$.
The product computation type $\prodtype{\ell \in L} \C_\ell$ classifies finite indexed products of computations, of types $\C_\ell$, for $\ell \in L$.
These tuples are not evaluated sequentially as in a call-by-value setting;
instead, a component of a tuple is evaluated only once it is selected by a projection.
Finally, the function type $A \to \C$ classifies computations of type $\C$ parametric on values of type $A$.

\subsection{Terms} \label{terms}

The terms of our language consist of \emph{value} terms~$V$, $W$, \emph{computation} terms~$M$, $N$, and \emph{handler} terms~$H$.
They are given by:
\begin{align*}
  V, W \bnfis {}
    &x \bnfor
    \f(V) \bnfor
    \one \bnfor
    \pair{V, W} \bnfor
    \ell(V) \bnfor
    \thunk M \\
  M, N \bnfis {}
    &\matchpair{V}{x, y}{M} \bnfor
    \matchcase{V}{\ell(x_\ell) \maps M_\ell}{\ell \in L} \bnfor
    \force V \bnfor \\
    &\ret V \bnfor
    \bind{M}{x \T A} N \bnfor
    \pair{M_\ell}_{\ell \in L} \bnfor
    \prj{\ell} M \bnfor
    \lam{x \T A} M \bnfor
    M \, V \bnfor \\
    &\op{V}(\abs{x \T \beta} M) \bnfor
    k(V) \bnfor
    \handleto{M}{H}{x \T A} N \\
  H \bnfis {}
    &\{ \op{x \T \alpha}(k \T \conttype{\beta}{\C}) \maps M_{\op*} \}_{\op* \T \optype{\alpha}{\beta}}
\end{align*}
where $x, y, \dots$ range over an assumed set of \emph{value variables}, and $k$ ranges over an assumed set of \emph{continuation variables}.
Here $\{ \mathord{\cdots} \}_{\ell \in L}$ represents a set of computations, one for each label $\ell \in L$;
similarly, $\{ \mathord{\cdots} \}_{\op* \T \optype{\alpha}{\beta}}$ represents a set of computations, one for each operation symbol $\op* \T \optype{\alpha}{\beta}$.

We may omit type annotations in bindings if it does not cause ambiguity;
we may also speak of values, computations, or handlers instead of value terms, computation  terms, or handler terms, respectively.

All the syntax is standard from call-by-push-value, other than that for the handler terms and the last line of the computation terms.
Value terms are built from value variables, the usual constructs for finite products, constructs for indexed sums, and thunked computations.
Note that value terms only involve constructors, for example, pairing $\pair{V, W}$,
  while the corresponding destructor terms, for example matching $\matchpair{V}{x, y}{M}$, are computations.

Next, there are constructor and destructor computation terms for computation types.
Perhaps the most interesting ones are the ones for the type~$F A$.
The constructor computation $\ret V$ returns the value $V$, while the sequencing construct $\bind{M}{x \T A} N$ evaluates $M$, binds the result to $x$, and proceeds as $N$.

\newcommand{\locvar}{l}

Next, there is an \emph{operation application} computation term:
\[
  \op{V}(\abs{x \T \beta} M)
\]
This first triggers the operation $\op*$ with parameter $V$ and then binds the outcome to $x$, proceeding as the \emph{continuation} $M$.
\begin{exas}
  \label{exa:operations} \hfill
  \begin{itemize}
    \item
      The computation
      \[
        \readop{\pair{}}(\abs{c \T \chr} \writeop{c}(\abs{x \T \unit} \writeop{c}(\abs{y \T \unit} \ret \one)))
      \]
      reads a character~$c$, entered by the user, prints it out twice, and returns the unit value.

    \item
      The computation
      \[
        \getop{\locvar}(\abs{n \T \nat} \setop{\pair{\locvar, n + 1}}(\abs{x \T \unit} \ret n))
      \]
      increments the number~$x$, stored in location~$\locvar$, and returns the old value.
  \end{itemize}
\end{exas}

\noindent Next, a handler term
\[
  \{ \op{x \T \alpha}(k \T \conttype{\beta}{\C}) \maps M_{\op*} \}_{\op* \T \optype{\alpha}{\beta}}
\]
is given by a finite set of \emph{handling operation definitions}
\[
  \op{x \T \alpha}(k \T \conttype{\beta}{\C}) \maps M_{\op*}
\]
one for each operation symbol~$\op*$.
The \emph{handling terms} $ M_{\op*}$ are dependent on their parameters, captured in their \emph{parameter variables}~$x$, and on the continuations of the handled operations, captured in their {continuation variables}~$k$.
Note that continuation variables never appear independently, only in the form~$k(V)$, where they are applied to a value~$V$.

\begin{exas}
  \label{exa:handlers} \hfill
  \begin{itemize}
    \item

      The exception handler
      \[
        \{ \raiseop{e}() \maps N_e \}_{e \in \exc}
      \]
      given in Section~\ref{sec:exception_handlers}, can be written as:
      \[
        \H{exc} = \{
          \raiseop{y \T \exc}(k \T \conttype{\void}{\C}) \maps \matchcase{y}{e(z) \maps N_e}{e \in \exc}
        \}
      \]
      where $y$ is the raised exception and $k$ is the continuation.
      As $\exc = \sumtype{e \in \exc} \unit$,
      we match $y$ against all possible cases $e(z)$, where $z$ is a dummy variable of type $\unit$.
      %
      %
      Note that we do not use the continuation~$k$ in the handling term.
      Indeed, $\raiseop*$ is a nullary operation symbol and there are no values of type $\void$ we could feed to the continuation, hence we cannot use it.

    \item
      Even though we cannot modify read-only state, we can still evaluate a computation with the state temporarily set to a different value.
      To do so, we use the \emph{temporary-state} handler that is dependent on a variable~$n \T \nat$; it is
      \[
        \H{temporary} = \{
          \getop{\locvar \T \loc}(k \T \conttype{\nat}{\C}) \maps k(n)
        \}
      \]
  \end{itemize} 
\end{exas}\medskip







\noindent Finally, we have the \emph{handling} computation term
\[
  \handleto{M}{H}{x \T A} N
\]
This evaluates the computation $M$, handling all operation application computations according to $H$, binds the result to $x$ and proceeds as $N$.
In more detail, handling works as follows.
Assume that $M$ triggers an operation application $\op{V}(\abs{y} M')$ and that the corresponding handling term is $\op{z}(k) \maps M_{\op*} \in H$.
Then, the operation is handled by evaluating $M_{\op*}$ instead, with the parameter variable $z$ bound to $V$ and with each occurrence of a term of the form $k(W)$ in $M_{\op*}$ replaced by
\[
  \handleto{M'[W / y]}{H}{x \T A} N
\]
Thus, the continuation $k$ receives an outcome $W$, determined by $M_{\op*}$, and is handled in the same way as $M$.
The handling term $M_{\op*}$ may use the continuation~$k$ any number of times and the behaviour of the handling construct can be very involved.
Note that while continuations are handled by $H$, the handling term $M_{\op*}$ itself is not.
Any operations it triggers or values it returns escape the handler.
They could however be handled by an enclosing handler.

We remark that sequencing $\bind{M}{x \T A} N$ is equivalent to the special case of handling in which $H$ handles all operations by themselves: see the discussion of this point in Section~\ref{sec:reasoning_about_handlers}.

\begin{exa}

  The simplest use for the handling construct is handling exceptions.
  Using handlers, we would write the computation
  \[
    \handleto{M}{\{ \raiseop{e}() \maps N_e \}_{e \in \exc}}{x \T A} N
  \]
  given in Section~\ref{sec:exception_handlers}, as
  %
  %
  %
  \begin{align*}
    \handleto{M}{\H{exc}}{x \T A} N
  \end{align*}
  %
%
  In this case, the behaviour matches the one given by Benton and Kennedy.
  If the computation $M$ returns a value $V$ then $N[V / x]$ is evaluated.
  If, instead, the computation $M$ triggers an exception $e$ then the replacement term $N_e$ is evaluated instead; so, if a value is then returned, that value is the final result of the entire computation and is not bound in $N$.
  %
\end{exa}

\begin{rem}
  The syntax of our handling construct differs from that of Benton and Kennedy:
  \[
    \kpre{try} x \T A \Leftarrow M \kop{in} N \kop{unless} \{ e \Rightarrow M_e \}_{e \in \exc}
  \]
  They noted some programming concerns regarding their syntax~\cite{benton01exceptional}.
  In particular, is not obvious that $M$ is handled but $N$ is not;
  this is especially the case when $N$ is large and the handler is obscured.
  An alternative they propose is:
  \[
    \kpre{try} x \T A \Leftarrow M \kop{unless} \{ e \Rightarrow M_e \}_{e \in \exc} \kop{in} N
  \]
  but then it is not obvious that $x$ is bound in $N$ but not in the handler.
  The syntax of our construct $\handleto{M}{H}{x \T A} N$ addresses both those issues.
  It also clarifies the order of evaluation:
    $M$ is handled with $H$ and its results are bound to $x$ and then used in $N$.
\end{rem}








In the next example  we use a standard  let binding abbreviation, defined as follows:
\[
  \letin{x \T A}{V} M \defeq (\lam{x \T A} M) \, V
\]

\begin{exa}
  Consider the temporary state handler $\H{temporary}$ given in Example~\ref{exa:handlers}.
  Then the computation
  \begin{align*}
    &\letin{n \T \nat}{20} \\
    &\handleto{\getop{\locvar}(\abs{x \T \nat} \getop{\locvar}(\abs{y \T \nat} \ret x + y))\\
    &}{\H{temporary}}{z \T A} \ret z + 2
  \end{align*}
  involving the stateful computation of Example~\ref{exa:operations} evaluates as follows:
    the first $\getop*$ operation is handled, with an outcome $20$ bound to $x$;
    the second $\getop*$ operation in the continuation is handled in the same way, again with an outcome $20$ bound to $y$;
    next, $\ret 20 + 20$ is handled with result $40$;
    this is substituted for $z$ in $\ret z + 2$, and the final result is $42$.
\end{exa}

As remarked in Section~\ref{sub:handling_arbitrary_effects},
  the language considered in~\cite{plotkin09handlers}  has restricted facilities for defining handlers.
In more detail, two levels of language are considered there.
In the first there are no handlers, and so no handling constructs.
The first level is used to define handlers,
  which, if they give models, are then used in handling constructs in the second level
  (which has no facilities for the further definition of handlers).
This can be considered a \emph{minimal} approach in contrast to that considered here,
  which can rather be considered \emph{maximal} as handlers and handling can be nested arbitrarily deeply.
The advantage of the maximal approach is that it accommodates all possible ways of treating the problem of ensuring that handlers give models.

\subsection{Typing judgements}

All typing judgements are made in \emph{value contexts}
\[
  \ctx = x_1 \T A_1, \dots, x_m \T A_m
\]
of value variables~$x_i$ bound to value types~$A_i$ and \emph{continuation contexts}
\[
  \kctx = k_1 \T \conttype{\alpha_1}{\C_1}, \dots, k_n \T \conttype{\alpha_n}{\C_n}
\]
of continuation variables~$k_j$ bound to \emph{continuation types} $\conttype{\alpha_j}{\C_j}$.
Continuation types have the form~$\conttype{\alpha}{\C}$ and type continuations $k$ that accept a value of an arity signature type $\alpha$ and proceed as a computation of type $\C$.
Values are typed as $\ctx \stoup \kctx \ent V \T A$,
  computations are typed as $\ctx \stoup \kctx \ent M \T \C$,
  and handlers are typed as $\ctx \stoup \kctx \ent H \T \handtype{\C}$.
Values are typed  according to the following rules:
\begin{mathpar}
  \deduct[x \T A \in \ctx]{}{
    \ctx \stoup \kctx \ent x \T A
  }

  \deduct[\f \T \alpha \to \beta]{
    \ctx \stoup \kctx \ent V \T \alpha
  }{
    \ctx \stoup \kctx \ent \f(V) \T \beta
  }

  \deduct{}{
    \ctx \stoup \kctx \ent \one \T \unit
  }

  \deduct{
    \ctx \stoup \kctx \ent V \T A \\
    \ctx \stoup \kctx \ent W \T B
  }{
    \ctx \stoup \kctx \ent \pair{V, W} \T A \times B
  }

  \deduct[\ell \in L]{
    \ctx \stoup \kctx \ent V \T A_\ell
  }{
    \ctx \stoup \kctx \ent \ell(V) \T \sumtype{\ell \in L} A_\ell
  }

  \deduct{
    \ctx \stoup \kctx \ent M \T \C
  }{
    \ctx \stoup \kctx \ent \thunk M \T U \C
  }
\end{mathpar}
Next, computations are typed according to the following rules:
\begin{mathpar}
  \deduct{
    \ctx \stoup \kctx \ent V \T A \times B \\
    \ctx, x \T A, y \T B \stoup \kctx \ent M \T \C
  }{
    \ctx \stoup \kctx \ent \matchpair{V}{x, y}{M} \T \C
  }

  \deduct{
    \ctx \stoup \kctx \ent V \T \sumtype{\ell \in L} A_\ell \\
    \ctx, x_\ell \T A_\ell \stoup \kctx \ent M_\ell \T \C \cond{\ell \in L}
  }{
    \ctx \stoup \kctx \ent \matchcase{V}{\ell(x_\ell) \maps M_\ell}{\ell \in L} \T \C
  }

  \deduct{
    \ctx \stoup \kctx \ent V \T U \C
  }{
    \ctx \stoup \kctx \ent \force V \T \C
  }

  \deduct{
    \ctx \stoup \kctx \ent V \T A
  }{
    \ctx \stoup \kctx \ent \ret V \T F A
  }

  \deduct{
    \ctx \stoup \kctx \ent M \T F A \\
    \ctx, x \T A \stoup \kctx \ent N \T \C
  }{
    \ctx \stoup \kctx \ent \bind{M}{x \T A} N \T \C
  }

  \deduct{
    \ctx \stoup \kctx \ent M_\ell \T \C_\ell \cond{\ell \in L}
  }{
    \ctx \stoup \kctx \ent \pair{M_\ell}_{\ell \in L} \T \prodtype{\ell \in L} \C_\ell
  }

  \deduct[\ell \in L]{
    \ctx \stoup \kctx \ent M \T \prodtype{\ell \in L} \C_\ell
  }{
    \ctx \stoup \kctx \ent \prj{\ell} M \T \C_\ell
  }

  \deduct{
    \ctx, x \T A \stoup \kctx \ent M \T \C
  }{
    \ctx \stoup \kctx \ent \lam{x \T A} M \T A \to \C
  }

  \deduct{
    \ctx \stoup \kctx \ent M \T A \to \C \\
    \ctx \stoup \kctx \ent V \T A
  }{
    \ctx \stoup \kctx \ent M \, V \T \C
  }

  \deduct[\op* \T \optype{\alpha}{\beta}]{
    \ctx \stoup \kctx \ent V \T \alpha \\
    \ctx, x \T \beta \stoup \kctx \ent M \T \C
  }{
    \ctx \stoup \kctx \ent \op{V}(\abs{x \T \beta} M) \T \C
  }

  \deduct[k \T \conttype{\alpha}{\C} \in \kctx]{
    \ctx \stoup \kctx \ent V \T \alpha
  }{
    \ctx \stoup \kctx \ent k(V) \T \C
  }

  \deduct{
    \ctx \stoup \kctx \ent M \T F A \\
    \ctx \stoup \kctx \ent H \T \handtype{\C} \\
    \ctx, x \T A \stoup \kctx \ent N \T \C
  }{
    \ctx \stoup \kctx \ent \handleto{M}{H}{x \T A} N \T \C
  }
\end{mathpar}
Finally, handlers are typed according to the following rule:
\[
  \deduct{
    \ctx, x \T \alpha \stoup \kctx, k \T \conttype{\beta}{\C} \ent
      M_{\op*} \T \C \cond{\op* \T \optype{\alpha}{\beta}}
  }{
    \ctx \stoup \kctx \ent \{\op{x \T \alpha}(k \T \conttype{\beta}{\C}) \maps M_{\op*} \}_{\op* \T \optype{\alpha}{\beta}} \T \handtype{\C}
  }
\]
Observe that $\kctx$ may contain more than one continuation variable when the handler being defined is used in handling definitions of other handlers.

\subsection{Abbreviations}
\label{sub:abbreviations}

Before we continue, let us introduce a few abbreviations to help make examples more readable.
First, we obtain arbitrary finite products from binary products:
\begin{align*}
  A_1 \times \dots \times A_n &\defeq (A_1 \times \dots \times A_{n - 1}) \times A_n \cond{n \geq 3} \\
  \pair{V_1, \dots, V_n} &\defeq \pair{\pair{V_1, \dots, V_{n - 1}}, V_n} \cond{n \geq 3}
\end{align*}
understanding binary product, where $n = 2$, as before,
the unit product, where $n = 1$ as simply $A_1$, and
the empty product, where $n = 0$, as $\unit$.

The main use of products is to pass around multiple values as one, so we set:
\begin{align*}
  \f(V_1, \dots, V_n) &\defeq \f(\pair{V_1, \dots, V_n}) \\
  \ell(V_1, \dots, V_n) &\defeq \ell(\pair{V_1, \dots, V_n}) \\
  \op{V_1, \dots, V_n}(\abs{x \T \beta} M) &\defeq
    \op{\pair{V_1, \dots, V_n}}(\abs{x \T \beta} M) \\
  k(V_1, \dots, V_n) &\defeq k(\pair{V_1, \dots, V_n}) \\
  \intertext{%
    Further, where possible, we omit empty parentheses in values and write:
  }
  \f &\defeq \f() \\
  \ell &\defeq \ell() \\
  k &\defeq k()
\end{align*}

We also adapt the tuple destructor to tuples of arbitrary finite size.
Using this destructor, we allow multiple variables in binding constructs such as sequencing or handler definitions.
For example, we set
\[
  \op{x_1, \dots, x_n}(k) \mapsto M
  \defeq
    \op{x}(k) \mapsto (\matchpair{x}{x_1, \dots, x_n} M)
\]
Similarly, we omit empty parentheses in binding constructs that bind no variables.
For the set of booleans, we set:
\[
  \ifThenelse{V}{M}{N} \defeq
    \matchcase{V}{\tru \maps M, \fls \maps N}{}
\]

We may use infix or other suitable notation when writing function applications.
We may assume additional function symbols if appropriate defining terms are available. Examples include logical function symbols, such as $\f[or] \T \bool \times \bool \to \bool$,  and   relations such as inequality on locations, $\not=_{\loc} \T \loc \times \loc \to \bool$.

For operation symbols~$\op* \T \optype{\alpha}{\type{n}}$, we define the usual finitary operation applications by:
\[
  \op{V}(M_1, \dots, M_n) \defeq \op{V}(
     \abs{x \T \type{n}} \matchcase{x}{i \maps M_i}{i \in \type{n}}
  )
\]
and write handling definitions as:
\[
  \op{x}(k_1, \dots, k_n) \maps M_{\op*}
\]
where in $M_{\op*}$, we write $k_i$ instead of $k(i)$ for $1 \leq i \leq n$.
In particular, we have:
\[
  \op{V}() \defeq \op{V}(\abs{x \T \void} \matchcase{x}{}{})
\]
and the handling definitions for nullary operations do not contain the corresponding continuation variable.
This agrees with the discussion given in Examples~\ref{exa:handlers}.






Operation applications can be somewhat cumbersome for writing programs, and instead we may write computations using \emph{generic effects}~\cite{plotkin03algebraic}.
The generic effect corresponding to an operation symbol~$\op* \T \optype{\alpha}{\beta}$ is defined by:
\[
  \gen{op} \defeq  \lam{x \T \alpha} \op{x}(\abs{y \T \beta} \ret y) \T \alpha \to F \beta
\]
Computations
\[
  \bind{\gen{op} \, V}{y \T \beta} M \qquad \text{and} \qquad \op{V}(\abs{y \T \beta} M)
\]
behave equivalently, in that both first trigger the operation with parameter~$V$, bind the outcome to~$x$, and proceed as~$M$.

\begin{exa}
  The computations of Example~\ref{exa:operations} could be written using generic effects as follows:
  \[
    \bind{\gen{read} \, \one}{c \T \chr} \gen{write} \, c \seq \gen{write} \, c
  \]
  and
  \[
    \bind{\gen{get} \, \locvar}{x \T \nat} \gen{set} \pair{\locvar, x + 1} \seq \ret x
  \]
  Here $M; N$ is the usual abbreviation for a sequencing $\bind{M}{x \T \unit} N$,
  where the useless result of $M$ is bound to a dummy variable $x$.
\end{exa}

When a handler term contains handling terms only for operation symbols from a subset~$\Theta$ of the set of operation symbols,
  we assume that the remaining operations are handled by themselves (so they are ``passed through").
Such a handler is defined by:
\[
  \{ \op{x}(k) \maps M_{\op*} \}_{\op* \in \Theta} \defeq
    \left\{ \op{x}(k) \maps \begin{cases}
      M_{\op*} & \cond{\op* \in \Theta} \\
      \op{x}(\abs{y \T \beta} k(y)) & \cond{\op* \notin \Theta}
    \end{cases} \right\}_{\op*}
\]

Sometimes we do not wish to write a value continuation in handlers.
Then, we use the following abbreviation:
\[
  \handle{M}{H} \defeq \handleto{M}{H}{x \T A} \ret x
\]
which employs a standard value continuation --- the identity one.
This abbreviation can be considered as a generalisation to arbitrary algebraic effects of the simple exception handling construct discussed in Section~\ref{sec:exception_handlers}.

There is a difference between
\[
  \handleto{M}{H}{x \T A} N
\]
which is the full handling construct,
\[
  \bind{(\handle{M}{H})}{x \T A} N
\]
which takes the result of the handled computation and binds it to $x$ in $N$, and
\[
  \handle{(\bind{M}{x \T A} N)}{H}
\]
which handles the computation that evaluates $M$ and binds the result to $x$ in $N$.
Both the first and the second computation handle only effects triggered by $M$, while the third computation handles effects triggered by both $M$ and $N$.
Furthermore, in the first computation, $x$ binds the value returned by $M$, while in the second computation, $x$ binds the value returned by $M$ once handled with $H$.

\begin{exa}
  To see the difference between the first two computations and the third, set:
  \begin{align*}
    H &\defeq \{
      \raiseop{e}() \maps \ret 10
    \} \\
    M &\defeq \ret 5 \\
    N & \defeq \raiseop{e}()
  \end{align*}
  Then the first two computations raise exception $e$, while the third one returns $10$.
  To see the difference between the second computation and the other two, set:
  \begin{align*}
    H &\defeq \{
      \raiseop{e}() \maps \ret 10
    \} \\
    M & \defeq \raiseop{e}() \\
    N & \defeq \ret 5
  \end{align*}
  Then the second computation returns $5$, while the other two return $10$.
\end{exa}

\section{Examples}
\label{sec:examples}

We now give some more examples to further demonstrate the scope of handlers of algebraic effects. As before such examples can be understood using the informal operational understanding of handlers given above. The question of the correctness of our examples is addressed 
in Remark~\ref{rem:correctness}, after the notions of effect theories and handler correctness have been presented.

\subsection{Explicit Nondeterminism}

\newcommand{\listb}{\type{list}_{\mathbf{b}}}
\newcommand{\failop}{\newop{fail}}
\newcommand{\listvar}{l}

The evaluation of a nondeterministic computation usually takes only one of all the possible paths.
An alternative is to take all the paths in some order  and allow the possibility of a path's failing.
This kind of nondeterminism is represented slightly differently from binary nondeterminism.
In addition to the binary operation symbol $\chooseop*$, we  take a nullary operation symbol $\failop*$, representing a path that failed.
This interpretation of nondeterminism corresponds to Haskell's nondeterminism monad~\cite{peyton-jones03haskell}.

We consider a handler which extracts the results of a computation into a list (which can then be operated on by other computations).
Since our calculus has no polymorphic lists --- although they could easily be added --- we limit ourselves to lists of a single base type $\mathbf{b}$.
We take a base type $\listb$ and 
function symbols:
  $\f[nil] \T \unit \to \listb$,
  $\f[cons] \T \mathbf{b} \times \listb \to \listb$,
  $\f[append] \T \listb \times \listb \to \listb$. 

Then, all the results of a computation $\ctx \stoup \kctx \ent M \T F \mathbf{b}$ can be extracted into a returned value of type $F \listb$ by
\[
  \ctx \stoup \kctx \ent \handleto{M}{\H{list}}{x \T \mathbf{b}} \ret \f[cons](x, \f[nil]) \T F \listb
\]
where $\H{list}$ is the handler given by:
\begin{align*}
   \ctx \stoup \kctx \ent {}&\{ \\
    &\quad \failop{}() \maps \ret \f[nil], \\
    &\quad \chooseop{}(k_1, k_2) \maps
      \bind{k_1}{\listvar_1 \T \listb}
      \bind{k_2}{\listvar_2 \T \listb}
      \ret \f[append](\listvar_1, \listvar_2) \\
  &\} \T \handtype{F \listb}
\end{align*}

\subsection{CCS}
\label{sub:ccs}

\newcommand{\nilop}{\newop{nil}}
\newcommand{\prefixop}{\newop{prefix}}
\newcommand{\CCSlabvar}{l}

To represent (the finitary part of) Milner's CCS~\cite{milner89-calculus} we take
  a type $\type{name}$ of (channel) names and an equality function symbol $=_{\type{name}} \T \type{name} \times \type{name} \to \type{name}$; we write $\type{lab}$ for a type of labels, abbreviating $\type{name}_{+} + \type{name}_{-}$, and $\type{act}$ for a type of actions, abbreviating $\unit_{\tau} + \type{lab}_{\text{lab}}$.
We further take
  three operation symbols representing combinators which we consider as effect constructors:
  deadlock $\nilop* \T \optype{\unit}{\void}$,
  action prefix $\prefixop* \T \optype{\type{act}}{\unit}$,
  and sum $\chooseop* \T \optype{\unit}{\type{2}}$.
We use the usual notation and write
  $P$ instead of $M$ for processes,
  $a.P$ instead of $\prefixop{a}(P)$, and
  $P_1 + P_2$ instead of $\chooseop{}(P_1, P_2)$.
Processes $P$ do not terminate normally, only in deadlock, hence we represent them as computations $P \T F \void$.

We consider the other CCS combinators as effect deconstructors.
Both relabelling and restriction can be represented using handlers. In order to write these handlers we assume available an equality function symbol $=_{\type{act}}\T \type{act} \times \type{act} \to \type{act}$ on actions and a ``dual" function symbol $\overline{\;\cdot\;}\T \type{lab} \to \type{lab}$: both have evident definitions. We also feel free to omit evident conversions, from $\type{name}$ to $\type{label}$, and from $\type{label}$ to $\type{act}$.

Relabelling $P[m / \CCSlabvar]$ replaces all actions  with label $\CCSlabvar$ (respectively $\overline{\CCSlabvar}$) in $P$ by  actions  with label $m$ (respectively $\overline{m}$);
it can be represented using the following handler:
\[\begin{array}{lcll}
  \ctx, \CCSlabvar \T \type{lab}, m \T \type{lab} \stoup \kctx & \ent   &  \{   a.k \maps &  \ifThenelse{a = \tau}{k}{\\
                    &  && \ifThenelse{a = \CCSlabvar}{m.k}{\\
                    &  && \ifThenelse{a = \overline{\CCSlabvar}}{\overline{m}.k}{a.k}}}\}\\
                    & & \;\T & \hspace{-35pt} \handtype{F \void}
\end{array}\]
One can deal with more general versions of renaming involving finitely given functions from labels to labels similarly.

Restriction $P \backslash n$ blocks all actions with name $n$  in $P$;
it can be represented using the following handler:
\[
  \ctx, n \T \type{name} \stoup \kctx \ent \{
     a.k \maps \ifThenelse{a = n \;\mathsf{or}\;  a = \overline{n} }{\nilop*()}{a.k}
      \} \T \handtype{F \void}
\]
One can deal with more general versions of restriction involving finite sets of names similarly.
The two handlers give handling terms only for $\prefixop*$
because our convention for handling omitted operations
gives exactly the expected structural behaviour of relabelling and restriction on $\nilop*$ and $\chooseop*$.

We do however, not know how to represent the final CCS combinator --- parallel, written $P \mid Q$.
Parallel is also an effect deconstructor, but unlike relabelling and restriction, which are both \emph{unary} deconstructors,
it is a \emph{binary} deconstructor as it reacts to actions of \emph{both} its arguments.
For a discussion of the difficulties in the treatment of such deconstructors, see~\cite
{glabbeek10on-cspb}.

The next few examples concern the use of parameter-passing handlers. We sometimes wish to handle different 
occurrences of the same operation differently, depending on the value of some parameter passed between the different occurrences.
Although each handler prescribes a fixed handling term for each operation,
  we can use handlers on function types $P \to \C$ to obtain $\C$  handlers that pass around parameters of value type $P$.
Each handling term then has the type $P \to \C$, rather than $\C$,  and
  captured continuations $k$ have the type $\conttype{\alpha}{(P \to \C)}$, rather than $\conttype{\alpha}{\C}$.

\subsection{Interactive input and output (I/O)}
Suppose we wish to suppress output after a certain number of characters have been printed out.
For any computation type $\C$, we define 
a 
character-suppressing handler $\H{suppress}$ by:
\begin{align*}
  \ctx, n_{\text{max}} \T \nat \stoup \kctx \ent {}&\{ \\
    &\quad \writeop{c}(k \T \unit \to (\nat \to \C)) \maps \lam{n \T \nat} \\
    &\qquad \ifThenelse{n < n_{\text{max}}}{\writeop{c}(k() \, (n + 1))}
    {k() \, n_{\text{max}}} \\
  &\} \T \handtype{\nat \to \C}
\end{align*}
The handling term for $\writeop*$ is dependent on $n \T \nat$, the number of characters printed out.
If this number is less than $n_{\text{max}}$, the maximum number of characters we want to print out,
  we write out the character and then handle the continuation, but now passing an incremented parameter to it.
But if $n \geq n_{\text{max}}$,
  we do not perform the $\writeop*$ operation, but continue with the handled continuation.
It does not matter exactly which parameter we pass to it, as long as it at least $n_{\text{max}}$.
Still, we call the continuation as it may return a value or trigger other operations.
(Using a convention introduced above  regarding operations of type $\optype{\alpha}{\type{n}}$, the definition of $\H{suppress}$ could have been written a little more elegantly, taking $k \T \nat \to \C$, and so on.) 

Since the type of the handler is not of the form $\handtype{F A}$,
  we need to specify a term for ``handling'' values in the handling construct, also dependent on the current value of the parameter.
For example, if we wish to handle $M \T F A$ with $\H{suppress}$, we write
\[
  \ctx \stoup \kctx \ent \handleto{M}{\H{suppress}}{x \T A}{\lam{n \T \nat} \ret x} \T \nat \to F A
\]
This means than no matter what the value of parameter $n$ is, we return the value $x$.
The handled computation has type $\nat \to F A$ and so, in order to obtain a computation of type $F A$,
  we need to apply it to the initial parameter $0 \T \nat$ as:
\[
  \ctx \stoup \kctx \ent (\handleto{M}{\H{suppress}}{x \T A}{\lam{n \T \nat} \ret x}) \, 0 \T F A
\]

\begin{rem}
  In the presence of parameters, the convention of handling omitted operations 
by themselves is still valid
  (this convention was used above for the input operation $\readop*$).
  What one wishes to do in the case of an operation that is not handled is to pass an unchanged parameter to the continuation, that is:
  \begin{align*}
    \op{y}(k) &\maps \lam{p \T P} \op{y}(\abs{x \T \beta} k(x) \, p) \\
    \intertext{%
      Since operations are defined pointwise on the function type, this is equivalent to
    }
    \op{y}(k) &\maps \op{y}(\abs{x \T \beta} \lam{p \T P} k(x) \, p) \\
    \intertext{%
      which, by $\eta$-equality, is equivalent to
    }
    \op{y}(k) &\maps \op{y}(\abs{x \T \beta} k(x))
  \end{align*}
  and is exactly what our convention assumes.
%
    The two equalities we have used here are 
     discussed further in Section~\ref{sec:reasoning_about_handlers}.
\end{rem}

\subsection{Timeout}

The \emph{timeout} handler furnishes another example of parameter-passing. 
This runs a computation and waits for a given amount of time.
If the computation does not complete in given time, it terminates it, returning a default value $x_0$ instead.

We represent time using a single operation
  $\delayop* \T \optype{\nat}{\unit}$,
where $\delayop{t}(M)$ is a computation that stalls for $t$ units of time, and then proceeds as $M$.
For any value type $A$, the timeout handler $\H{timeout}$ is given by:
\begin{align*}
  \ctx, x_0 \T A, t_{\text{wait}} \T \nat \stoup \kctx \ent {}&\{ \\
    &\quad\delayop{t}(k \T \nat \to F A) \maps \lam{t_{\text{spent}} \T \nat} \\
      &\qquad\ifThenelse{t + t_{\text{spent}} \leq t_{\text{wait}} \\&\qquad}
        {\delayop{t}(k(t + t_{\text{spent}}))\\&\qquad}
        {\delayop{t_{\text{wait}} - t_{\text{spent}}}(\ret x_0)} \\
  &\} \T \handtype{\nat \to F A}
\end{align*}
(now making use of the above convention introduced regarding operations $\op* \T \optype{\alpha}{\type{n}}$). 
The handler is used on a computation $M \T F A$ as follows:
\[
  \ctx, x_0 \T A, t_{\text{wait}} \T \nat \stoup \kctx \ent (\handleto{M}{\H{timeout}}{x \T A} \lam{t \T \nat} \ret x) \, 0 \T F A
\]
Note that the handling term preserves the time spent during the evaluation of the handled computation.

\subsection{Rollback}

 When a computation raises an exception while modifying the memory, for example,
  when a connection drops halfway through a database transaction, we may want to revert all modifications made during the computation.
This behaviour is termed \emph{rollback}.

Assuming, for the sake of simplicity, that there is only a single location, which is 
given by a term $\locvar_0$, 
an appropriate rollback handler $\H{rollback}$ is given by:
\[
  \ctx, n_{\text{init}} \T \nat \stoup \kctx \ent \{
    \raiseop{e}(k) \maps \setop{\locvar_0, n_{\text{init}}}(M' \, e)
  \} \T \handtype{\C}
\]
where $M'\T \exc \to \C$. To evaluate a computation $M$, rolling back to the initial state $n_{\text{init}}$ in the case of exceptions, we write:
\[
  \ctx \stoup \kctx \ent \getop{\locvar_0}(\abs{n_{\text{init}} \T \nat} \handle{M}{\H{rollback}}) \T \C
\]

An alternative is a parameter-passing handler, which does not modify the memory,
  but keeps track of all the changes to the location $\locvar_0$ in the parameter $n$.
Then, once the handled computation has returned a value, meaning that no exceptions have been raised, the parameter is committed to the memory.
This handler $\H{param-rollback}$ is:
\begin{align*}
  \ctx \stoup \kctx \ent {}&\{ \\
  &\quad\getop{\locvar_0}(k \T \conttype{\nat}{(\nat \to \C)}) \maps \lam{n \T \nat} k(n) \, n \\
  &\quad\setop{\locvar_0, n'}(k \T \conttype{\unit}{(\nat \to \C)}) \maps \lam{n \T \nat} k() \, n' \\
  &\quad\raiseop{e}() \maps \lam{n \T \nat} M' \, e \\
  &\} \T \handtype{\nat \to \C}
\end{align*}
It is used on a computation $M$ as follows:
\[
  \getop{\locvar_0}(\abs{n_{\text{init}} \T \nat}
    (
      \handleto{M}{\H{param-rollback}}{x \T A} \lam{n \T \nat} \setop{\locvar_0,n}(\ret x)
    ) \, n_{\text{init}}
  )
\]
Here the initial state is read, obtaining $n_{\text{init}}$, which is passed to the handler as the initial parameter value.
If no exception is raised and a value $x$ is returned, the state is updated to reflect the final value $n$ of the parameter.

\subsection{Stream redirection}

Let us conclude with a practical example of processes in a \textsc{Unix}-like operating system.
These processes 
read their input and write their output through standard input and output channels.
These channels are usually connected to a keyboard and a terminal window.
However, an output of one process can be piped to the input of another one, allowing multiple simple processes to be combined into more powerful ones.

We present a simplified model of files and devices.  We take a base type~$\type{file}$ to represent files and two operation symbols:
  $\readop* \T \optype{\type{channel}}{\chr}$ and $\writeop* \T \optype{\chr \times \type{channel}}{\unit}$.
Here, $\type{channel}$ is an abbreviation for the sum $\unit_{\text{std}} + \type{file}_{\text{file}}$, where $\text{std}$ represents the standard input and output channel, and $\text{file}(f)$ represents the file~$f$.
A more realistic 
model would have to include an operation that allows one to open a file and thus obtain a needed channel, and a similar operation for closing an opened file. In our model, one may instead assume that all files have already been opened for reading and writing.

We begin with the redirection \texttt{p > out} which takes the output stream of a process \texttt{p} and writes it to a file \texttt{out}.
This is used to either automatically generate files or to log the activity of processes.
The redirection is written as:
\[
  \ctx \stoup \kctx \ent \letin{f \T \type{file}}{\mathtt{out}} \handle{\mathtt{p}}{\H{\texttt{>}}} \T \C
\]
where the handler~$\H{\texttt{>}}$ is given by:
\begin{align*}
  \ctx, f \T \type{file} \stoup \kctx \ent {} &\{ \\
    &\quad\writeop{c, ch}(k) \maps \matchcase{ch}{ \\
      &\qquad \text{std} \maps \writeop{c, \text{file}(f)}(k), \\
      &\qquad \text{file}(f') \maps \writeop{c, \text{file}(f')}(k) \\
    &\quad}{} \\
  &\} \T \handtype{\C}
\end{align*}

\textsc{Unix} allows some devices to present themselves as ordinary files.
This, for example, allows a program to print out a document by simply redirecting its output to a file that corresponds to the printer.
A general treatment of 
 devices is beyond the scope of this paper, but some simple devices can be modelled with handlers.

An example is the null device \texttt{/dev/null}, which discards everything that is written to it.
The command \texttt{p > /dev/null} hence effectively suppresses the output of the process \texttt{p}.
The same behaviour can be achieved using the handler~$\H{\texttt{>\,/dev/null}}$, given by:
\begin{align*}
  \ctx \stoup \kctx \ent {}&\{ \\
    &\quad\writeop{c, ch}(k) \maps \matchcase{ch}{ \\
      &\qquad \text{std} \maps k, \\
      &\qquad \text{file}(f') \maps \writeop{c, \text{file}(f')}(k) \\
    &\quad}{} \\
  &\} \T \handtype{\C}
\end{align*}

A similar redirection \texttt{p < in} reads the file \texttt{in} and passes its contents to the process \texttt{p}.
This redirection can be represented using a handler that now replaces the standard input with a given file in all $\readop*$ operations.
It is given by:
\begin{align*}
  \ctx, f \T \type{file} \stoup \kctx \ent {}&\{ \\
    &\quad\readop{ch}(k) \maps \matchcase{ch}{ \\
      &\qquad \text{std} \maps \readop{\text{file}(f)}(\abs{c \T \chr} k(c)), \\
      &\qquad \text{file}(f') \maps \readop{\text{file}(f')}(\abs{c \T \chr} k(c))) \\
    &\quad}{} \\
  &\} \T \handtype{\C}
\end{align*}
Both redirections can be combined so \texttt{(p < in) > out} reads the input file and writes the processed contents to the output file.

\newcommand{\pipe}{\!\!|\!\!\ }

We next consider \textsc{Unix} pipes \texttt{p1 \pipe p2}, where the output of \texttt{p1} is fed to the input of \texttt{p2}.
Using handlers we can express simple cases of the pipe combinator  \texttt{|}.
For example, consider the pipe \texttt{yes \pipe p},
  where the process \texttt{yes} outputs an infinite stream made of a predetermined character (the default one being \texttt{y}).
Such a pipe then gives a way of routinely confirming a series of actions, for example deleting a large number of files.
(This is not always the best way, since processes usually provide a safer means of doing the same thing, but is often useful when they do not.)
This particular pipe may be written using the following handler:
\[
  \ctx, c \T \chr \stoup \kctx \ent \{
    \readop{\text{std}}(k) \maps k(c)
  \} \T \handtype{\C}
\]
The general pipe combinator \texttt{p1 \pipe p2} is harder to represent, because like the CCS parallel combinator, and unlike the above redirections and the simple pipe, it is a binary rather than a unary deconstructor.

It would be very interesting to investigate to what extent the effects
%
supported by \textsc{Unix}
can be realistically modelled using the methods of the algebraic theory of effects: as well as files and devices there are, for example, signals or process scheduling.

\section{Semantics}
\label{sec:semantics}


We now introduce effect theories and their interpretations and then give the denotational semantics of our language for handling algebraic effects. For the sake of simplicity, we largely limit ourselves to sets, but show in Section~\ref{sec:recursion} how everything adapts straightforwardly to $\omega$-cpos. 
%


\subsection{Effect theories}
\label{sub:effect_theory}

We describe properties of effects with equations between \emph{templates}~$T$.
These describe the general shape of all computations, regardless of their type.
Assuming a given signature, templates are given by:
\[
  T \bnfis
    z(V) \bnfor
    \matchpair{V}{x, y}{T} \bnfor
    \matchcase{V}{\ell(x_\ell) \maps T_\ell}{\ell \in L} \bnfor
    \op{V}(\abs{x \T \beta} T)
\]
where $z$ ranges over a given set of \emph{template variables}.

In templates, we limit ourselves to \emph{signature values}.
These are values that can be typed as $\ctx \ent V \T \alpha$, where $\ctx$ is a context of value variables bound to signature types;
the typing rules for this judgement are, with the omission of continuation contexts and the rule for typing thunks, the same as those for values.

We build templates in a context of value variables, bound to signature types, and  a \emph{template context}
\[
  \tctx = z_1 \T \alpha_1, \dots, z_n \T \alpha_n
\]
of template variables~$z_j$, bound to arity signature types $\alpha_j$.
Note: $z_j \T \alpha_j$ does not represent a value of type~$\alpha_j$, but a computation dependent on such a value.
As templates describe common properties of computations, we do not assign them a type,
  but only check whether they are well-formed, relative to a value context and a template context.
The typing judgement for being  a well-formed template is $\ctx \stoup \tctx \ent T$; it is given by the following rules:
\begin{mathpar}
  \deduct[z \T \alpha \in \tctx]{
    \ctx \ent V \T \alpha
  }{
    \ctx \stoup \tctx \ent z(V)
  }

  \deduct{
    \ctx \ent V \T \alpha \times \beta \\
    \ctx, x \T \alpha, y \T \beta \stoup \tctx \ent T
  }{
    \ctx \stoup \tctx \ent \matchpair{V}{x, y}{T}
  }

  \deduct{
    \ctx \ent V \T \sumtype{\ell \in L} \alpha_\ell \\
    \ctx, x_\ell \T \alpha_\ell \stoup \tctx \ent T_\ell \cond{\ell \in L}
  }{
    \ctx \stoup \tctx \ent
       \matchcase{V}{\ell(x_\ell) \maps T_\ell}{\ell \in L}
  }

  \deduct[\op* \T \optype{\alpha}{\beta}]{
    \ctx \ent V \T \alpha \\
    \ctx, x \T \beta \stoup \tctx \ent T
  }{
    \ctx \stoup \tctx \ent \op{V}(\abs{x \T \beta} T)
  }
\end{mathpar}

An \emph{effect theory}~$\thy$ over a given signature 
is a finite set of equations
  $\ctx \stoup \tctx \ent T_1 = T_2$,
where $T_1$ and $T_2$ are well-formed relative to $\ctx$ and $\tctx$.

Some examples follow, continuing on from the corresponding signatures given in Example~\ref{exa:effects} and Section~\ref{sec:examples}.
We make use of abbreviations  for handling tuples, etc.,  in templates that are analogous to those introduced above for the other kinds of terms.
\goodbreak
\begin{exas}
  \nopagebreak
  \label{exa:effect_theories} \hfill
  \nopagebreak
  \begin{description}
  \setlength{\itemsep}{\bigskipamount}
  \item[Exceptions]
    The effect theory is the empty set, as exceptions satisfy no nontrivial 
    equations.
  \item[State]
    The effect theory consists of the following equations (for readability, we write this and other theories without contexts):
    \begin{align*}
        \getop{\locvar}(\abs{x} z) &= z \\
        \getop{\locvar}(\abs{x} \setop{\locvar, x}(z)) &= z \\
      \setop{\locvar, x}(\setop{\locvar, x'}(z)) &= \setop{\locvar, x'}(z) \\
      \setop{\locvar, x}(\getop{\locvar}(\abs{x'} z(x')) &= \setop{\locvar, x}(z(x)) \\
      \getop{\locvar}(\abs{x} \getop{\locvar}(\abs{x'} z(x, x'))) &=
        \getop{\locvar}(\abs{x} z(x, x)) \\
      \setop{\locvar, x}(\setop{\locvar', x'}(z)) &= \setop{\locvar', x'}(\setop{\locvar, x}(z)) &\cond{\locvar \neq \locvar'} \\
      \setop{\locvar, x}(\getop{\locvar'}(\abs{x'} z(x')) &=
        \getop{\locvar'}(\abs{x'} \setop{\locvar, x}(z(x')))
        &\cond{\locvar \neq \locvar'} \\
      \getop{\locvar}(\abs{x} \getop{\locvar'}(\abs{x'} z(x, x'))) &=
        \getop{\locvar'}(\abs{x'} \getop{\locvar}(\abs{x} z(x, x'))
        &\cond{\locvar \neq \locvar'}
    \end{align*}
    We found it convenient to write the last three equations with a side condition $\locvar \neq \locvar'$.
    This still remains within the scope of our definition of an effect theory,
    reading $T_1 = T_2 \;(\locvar \neq \locvar')$ as an abbreviation of the equation
    \[
      T_1 = (\ifThenelse{\locvar \not =_{\loc}  \locvar'}{T_2}{T_1})
    \]
  There is some redundancy in these equations, with (see~\cite{plotkin02notions,Mellies10}) the first, fifth, and eight being consequences of the others for a suitable notion of equational consequence (a semantic such notion can be provided using the interpretations of effect theories defined in Section~\ref{sec:int-eff-ths} below).

\item[Read-only state]
    The effect theory consists of the following equations:
    \begin{align*}
      \getop{\locvar}(\abs{x} z) &= z \\
      \getop{\locvar}(\abs{x} \getop{\locvar}(\abs{x'} z(x, x'))) &=
        \getop{\locvar}(\abs{x} z(x, x)) \\
      \getop{\locvar}(\abs{x} \getop{\locvar'}(\abs{x'} z(x, x'))) &=
        \getop{\locvar'}(\abs{x'} \getop{\locvar}(\abs{x} z(x, x'))
        &\cond{\locvar \neq \locvar'}
      \end{align*}
   \item[Nondeterminism]
    The effect theory consists of the following equations:
    \begin{align*}
      \chooseop*(z, z) &= z \\
      \chooseop*(z_1, z_2) &= \chooseop*(z_2, z_1) \\
      \chooseop*(\chooseop*(z_1, z_2), z_3) &=
        \chooseop*(z_1, \chooseop*(z_2, z_3))
    \end{align*}
  \item[Explicit nondeterminism]
    The effect theory consists of the following equations:
    \begin{align*}
      \chooseop*(\chooseop*(z_1, z_2), z_3) &=
        \chooseop*(z_1, \chooseop*(z_2, z_3))\\
      \chooseop*(z, \nilop*()) &= z \\
      \chooseop*(\nilop*(), z) &= z
    \end{align*}
    \item[I/O, Stream redirection]
      The effect theory is (again) the empty set.
  \item[Time]
    The effect theory consists of the following equations:
    \begin{align*}
      \delayop{0}(z) &= z \\
      \delayop{t_1}(\delayop{t_2}(z)) &= \delayop{t_1 + t_2}(z)
    \end{align*}
  \item[CCS]
    We take the effect theory for CCS to be the 
    union of the effect theories for nondeterminism and explicit nondeterminism, i.e., it consists of the following equations:
     \begin{align*}
      \chooseop*(z, z) &= z \\
      \chooseop*(z_1, z_2) &= \chooseop*(z_2, z_1) \\
      \chooseop*(\chooseop*(z_1, z_2), z_3) &=
        \chooseop*(z_1, \chooseop*(z_2, z_3))\\
        \chooseop*(z, \nilop*()) &= z \\
      \chooseop*(\nilop*(), z) &= z
    \end{align*}
    These equations are those axiomatising strong bisimulation; for weak bisimulation one would add Milner's $\tau$ laws~\cite{milner89-calculus}.
    %
  \item[Destructive exceptions]
    The signature of destructive exceptions is that of state together with that of exceptions, and its theory
    consists of all the equations for state and the following two equations:
    \begin{align*}
      \getop{\locvar}(\abs{x} \raiseop{e}()) &= \raiseop{e}() \\
      \setop{\locvar, x}(\raiseop{e}()) &= \raiseop{e}()
    \end{align*}
    As $\raiseop*$ is nullary, we have written $\raiseop{e}()$ instead of $\raiseop{e}(\abs{x \T \void} T)$,
      emphasising that the continuation never gets evaluated.

    The equations imply that a memory operation followed by raising an exception is the same as just raising the exception.
    Effectively, this implies that all memory operations are moot if an exception occurs, hence the terminology of  destructive exceptions.
   Observe that the first equation is an instance of one occurring in the effect theory for state.
%
  The theory of destructive exceptions is discussed as that of ``rollback" in~\cite{hyland06combining}, where it is given as an example of the \emph{tensor} of two theories.

\end{description}
\end{exas}

\subsection{Interpreting effect theories} \label{sec:int-eff-ths}

We assume given an effect theory $\thy$.
We begin by interpreting each signature type $\alpha$ by a set $\itp{\alpha}$.

Suppose we are given an assignment $\itp{\mathbf{b}}$ of a set to each base type $\mathbf{b}$,
  such that $\itp{\mathbf{b}}$ is countable in the case that $\mathbf{b}$ is an arity base type.
Finite products and sums of signature types are then interpreted using the corresponding set operations,
  and base types are interpreted by the assigned sets.
Note that $\itp{\alpha}$ is then a countable set for any arity signature type $\alpha$;
this will allow us to interpret effects in terms of countable equational theories.

An \emph{interpretation} of $\thy$ then consists of such an assignment together with a map
  $\itp{\f} \from \itp{\alpha} \to \itp{\beta}$
for each function symbol~$\f \T \alpha \to \beta$. The intended interpretations for each of our examples should be clear. For instance $\itp{\mathbf{\listb}}$ is the set of lists of 
elements of $\itp{\mathbf{b}}$,
$\itp{\f[append] } \from \itp{\listb} \to \itp{\listb}$ is list concatenation, and  $\itp{\type{name}}$ is a given set with equality function $\itp{=_{\type{name}}} \T \itp{\type{name}} \times \itp{\type{name}}  \to \itp{\type{name}}$.


We assume given such an interpretation.
We can then interpret each well-typed signature value $\ctx \ent V \T \alpha$ in an evident way by a map
  $\itp{V} \from \itp{\ctx} \to \itp{\alpha}$,
where value contexts $\ctx$ are interpreted component-wise by setting
  $\itp{x_1 \T \alpha_1, \dots, x_m \T \alpha_m} =
  \itp{\alpha_1} \times \dots \times \itp{\alpha_m}$.

Take a set $|\model|$ and an \emph{operation}
\[
  \op{\model} \from \itp{\alpha} \times |\model|^{\itp{\beta}} \to |\model|
\]
for each operation symbol $\op* \T \optype{\alpha}{\beta}$.
Then, we can interpret a well-formed template $\ctx \stoup \tctx \ent T$ by a map
  $\itp{\ctx \stoup \tctx \ent T} \from \itp{\ctx} \times \itp{\tctx} \to |\model|$,
where template contexts $\tctx$ are interpreted component-wise by
  $\itp{z_1 \T \alpha_1, \dots, z_n \T \alpha_n} =
  |\model|^{\itp{\alpha_1}} \times \dots
  \times |\model|^{\itp{\alpha_n}}$.

We can then interpret templates:
 %
 %
 %
 \begin{align*}
  \itp{\ctx \stoup \tctx \ent z_j(V)}(a, b) &=
    \Pr j(b)(\itp{V}(a)) \cond{1 \leq j \leq n}\\
   \itp{\ctx \stoup \tctx \ent \matchpair{V}{x, y}{T}}(a, b) &=
    \itp{T}(\pair{a, \Pr 1(\itp{V}(a)), \Pr 2(\itp{V}(a))}, b) \\
  \itp{\ctx \stoup \tctx \ent \matchcase{V}{\ell(x_\ell) \maps T_\ell}{\ell \in L}}(a, b) &=
    \itp{T_\ell}(\pair{a, x}, b) \cond{\text{if $\itp{V}(a) = \ell(x)$}} \\
  \itp{\ctx \stoup \tctx \ent \op{V}(\abs{x \T \beta} T)}(a, b) &=
    \op{\model}(\itp{V}(a), \itp{T}(\pair{a, -}, b))
\end{align*}
where on the right,
  $\Pr i \from A_1 \times \dots \times A_n \to A_i$
    is the $i$-th projection,
and we have abbreviated evident typing judgements. 
We also used a standard notation for functions, writing  $\itp{T}(\pair{a, -}, b)$ for the function 
$x \T \itp{\beta} \mapsto \itp{T}(\pair{a, x}, b)$.


\begin{defi}
  A \emph{model}~$\model$ of the effect theory~$\thy$ is a set $|\model|$,
    called the \emph{carrier} of the model,
    together with an \emph{operation}
  \[
    \op{\model} \from
      \itp{\alpha} \times |\model|^{\itp{\beta}} \to |\model|
  \]
  for each operation symbol $\op* \T \optype{\alpha}{\beta}$,
    such that $\itp{T_1} = \itp{T_2}$ holds for all the equations $\ctx \stoup \tctx \ent T_1 = T_2$ in $\thy$.

  A \emph{homomorphism}~$\homo \from \model_1 \lol \model_2$ is a map $ \homo \from |\model_1| \to |\model_2|$ such that
  \[
     \homo \circ \op{\model_1} =
    \op{\model_2} \circ
      (\Id[\itp{\alpha}] \times \homo^{\itp{\beta}})
  \]
  holds for all operation symbols $\op* \T \optype{\alpha}{\beta}$.

  The models of $\thy$ and the homomorphisms between them form a category $\Mod[\thy]$.
  This category is equipped with a forgetful functor $U \from \Mod[\thy] \to \Set$,  where
  $U\model \defeq  |\model|$ and $Uh \defeq h$.
\end{defi}

\begin{prop} \label{prop:free}
  The functor $U \from \Mod[\thy] \to \Set$ has a left adjoint $F \from \Set \to \Mod[\thy]$.
\end{prop}

\begin{proof}
  The existence of the adjoint functor follows from the corresponding result in the context of countable equational theories~\cite{gratzer79universal}.

  \newcommand{\thyc}{\thy_\omega}
  First we obtain a countable equational theory $\thyc$ from $\thy$.
  The signature of $\thyc$ consists of
    operations $\op*_a$ of arity $|\itp{\beta}|$
    for each operation symbol $\op* \T \optype{\alpha}{\beta}$
    and each $a \in \itp{\alpha}$.
  As $\beta$ is an arity signature type, $\itp{\beta}$ is a countable set, hence each operation has a countable arity.

  Next, for each well-formed template
  \[
   \ctx \stoup z_1 \T \alpha_1, \dots, z_n \T \alpha_n \ent T
  \]
  and each $c \in \itp{\ctx}$, we get a $\thyc$ term $\tctx' \T T^c$,
    where the context $\tctx'$ consists of variables $z_i^a$ for each $1 \leq i \leq n$ and each $a \in \itp{\alpha_i}$.
  Since all the $\alpha_i$ are again arity signature types,
    the sets $\itp{\alpha_i}$ are countable, and so $\tctx$ has countably many variables.
  The term $T^c$ is defined recursively by:
  \begin{align*}
    (z_i(V))^c &= z_i^{\itp{V}(c)} \\
    (\matchpair{V}{x, y}{T})^c &= T^{\pair{c,\Pr1(\itp{V}(c)), \Pr2(\itp{V}(c))}} \\
    (\matchcase{V}{\ell(x_\ell) \maps T_\ell}{\ell \in L})^c &=
      T_{\ell'}^{\pair{c, a}} \cond{\itp{V}(c) = \In{\ell'}(a)} \\
    (\op{V}(\abs{x \T \beta} T))^c &= \op{\itp{V}(c)}(T^{\pair{c, b}})_{b \in \itp{\beta}}
  \end{align*}
  We then take $\thyc$ to be the equational theory generated by all equations $\tctx' \ent T_1^c = T_2^c$,
    where $\ctx \stoup \tctx \ent T_1 = T_2$ is a $\thy$ equation and $c \in \itp{\ctx}$.
  It is not difficult to see that the category of models of $\thyc$ is equivalent to $\Mod[\thy]$,
    with the equivalence being consistent with the forgetful functors;
  the result follows.
\end{proof}

\newcommand{\dagsub}[2]{#1^{\dagger_{#2}}}
The left adjoint $F$ is called the \emph{free model functor};
for a set~$A$, the model $F A$ is called the \emph{free model over~$A$};
and for a map $f \from A \to U \model$, the adjoint homomorphism $\dagsub{f}{\model} \from F A \lol \model$ is called the \emph{homomorphism induced by $f$}.
We write $\eta_A \from A \to FA$ for the unit, as usual.

\subsection{Interpreting values and computations}

We assume given
  a signature (so that value and computation terms are determined),
  an effect theory $\thy$ over that signature, and
  an interpretation of the effect theory
  (so that the category of models, etc., is determined).

We can then interpret value types $A$ by sets $\itp{A}$ and computation types $\C$ by models $\itp{\C}$ of the given effect theory~$\thy$.
Value types are interpreted in the same way as signature types,
  except for $U \C$, which is interpreted as the carrier of $\itp{\C}$.
The computation type $F A$ is interpreted by the free model on $\itp{A}$.
Products of computation types are interpreted by product models and function types are interpreted by exponent models.
These are defined as follows:
\begin{defi}
  For any family of models $\model_\ell$, where $\ell$ ranges over the index
  set~$L$, the \emph{product model}~$\prod_{\ell \in L} \model_\ell$ is
   the model with carrier $\prod_{\ell \in L} |\model_\ell|$ and
  component-wise defined operations.
  For any model $\model$ and any set $A$, the
  \emph{exponent model}~$\model^A$ is the model with
  carrier $| \model|^A$ and pointwise defined operations.
\end{defi}
\newcommand{\eqdef}{=_{\mathrm{def}}}
\newcommand{\pto}{\rightharpoonup}
Contexts are again interpreted component-wise,
  with $\ctx = x_1 \T A_1, \dots, x_m \T A_m$ being interpreted by
  $\itp{\Gamma} \eqdef \itp{A_1} \times \dots \times \itp{A_m}$,
  and
  $\kctx = k_1 \T \conttype{\beta_1}{\C_1}, \dots, k_n \T \conttype{\beta_n}{\C_n}$ being interpreted by
  $\itp{\kctx} \eqdef U \itp{\C_1}^{\itp{\beta_1}} \times \dots \times U \itp{\C_n}^{\itp{\beta_n}}$.

We would like 
to interpret
  well-typed value terms $\ctx \stoup \kctx \ent V \T A$ by maps of the form
  \[\itp{\ctx \stoup \kctx \ent V \T A} \from \itp{\ctx} \times \itp{\kctx} \to \itp{A}\]
  well-typed computation terms $\ctx \stoup \kctx \ent M \T \C$ by maps of the form 
  \[\itp{\ctx \stoup \kctx \ent M \T \C} \from \itp{\ctx} \times \itp{\kctx} \to U \itp{\C}\]
  and well-typed handler terms $\ctx \stoup \kctx \ent H \T \handtype{\C}$ by maps of the form
  \[
    \itp{\ctx \stoup \kctx \ent H \T \handtype{\C}} \from \itp{\ctx} \times \itp{\kctx}
    \to
    \prod_{\op* \T \optype{\alpha}{\beta}} \itp{\alpha}
    \times
    U \itp{\C}^{\itp{\beta}} \rightarrow U \itp{\C}
  \]
  such that, for every $c \in  \itp{\ctx} \times \itp{\kctx} $, the set $U \itp{\C}$,
  together with the operations $\Pr {\op*}(\itp{H}(c))$, forms a model of  $\thy$.

However, as discussed informally above, not all handlers  receive an interpretation.
As handler terms occur in computation terms and computation terms occur in value terms,
  computation terms and value terms may also not receive an interpretation.
  So, instead, we interpret  
  well-typed value terms $\ctx \stoup \kctx \ent V \T A$ by partial maps 
    \[\itp{\ctx \stoup \kctx \ent V \T A} \from \itp{\ctx} \times \itp{\kctx} \pto \itp{A}\]
  on $ \itp{\ctx} \times \itp{\kctx}$, and similarly for computation and handler terms.

 Fix contexts
  $\ctx = x_1 \T A_1, \dots, x_m \T A_m$ and
  $\kctx = k_1 \T \conttype{\beta_1}{\C_1}, \dots, k_n \T \conttype{\beta_n}{\C_n}$.
Value terms are  interpreted as follows:
\begin{align*}
  \itp{\ctx \stoup \kctx \ent x_i \T A_i}(a, b)  &\simeq  \Pr i(a)  \cond{1 \leq i \leq m}\\
  \itp{\ctx \stoup \kctx \ent \f(V) \T \beta}(c)  &\simeq
    \itp{\f}(\itp{V}(c))
    \cond{\f \T \alpha \to \beta} \\
  \itp{\ctx \stoup \kctx \ent \one \T \unit}(c)  &\simeq
   \ast\\ 
  \itp{\ctx \stoup \kctx \ent \pair{V, W} \T A \times B}(c)  &\simeq
    \pair{\itp{V}(c) , \itp{W}(c) } \\
  \itp{\ctx \stoup \kctx \ent \ell(V) \T \sumtype{\ell \in L} A_\ell}(c)  &\simeq
    \In \ell(\itp{V}(c)) \\
  \itp{\ctx \stoup \kctx \ent \thunk M \T U \C}(c) &\simeq \itp{M}(c)
\end{align*}
In the above, we have used  \emph{Kleene equality $ e \simeq e'$}
  which holds if either both the expressions $e$ and $e'$ are defined and equal, or both are undefined.
Further,  typing judgements are abbreviated on the right-hand side,
  $\ast$ is the unique element of the one-element set $\mathbbm{1}$
  and $\In \ell \from A_\ell \to \sum_{\ell \in L} A_\ell$ is the $\ell$-th injection into a disjoint sum.

Next, computation terms, other than the handling construct (for which see below) are interpreted as follows:
\begin{align*}
  %
  %
  &\itp{\ctx \stoup \kctx \ent \matchpair{V}{x, y}{M} \T \C}(a, b) \\
  &\quad\simeq
    \begin{cases}  \itp{M}(\pair{a, \Pr 1( \itp{V}(a, b)), \Pr 2 (\itp{V}(a, b))}, b) & (\text{if }\forall x\in \itp{A},y \in \itp{B}.\, \itp{M}(\pair{a,x,y},b)\downarrow )\\
                            \text{undefined}              & (\text{otherwise})
    \end{cases}\\ 
  &\quad(\text{where $\ctx \stoup \kctx \ent V \T A \times B$ and $\ctx, x \T A, y \T B \stoup \kctx \ent M \T \C$}) \\[1ex]
        %
  &\itp{\ctx \stoup \kctx \ent \matchcase{V}{\ell(x_\ell) \maps M_\ell}{\ell \in L} \T \C}(c) \\
  &\quad\simeq
    \begin{cases}
      \itp{M_{\ell}(\pair{a, x},b)} & (\text{if $\itp{V}(a,b) \simeq \ell(x)$ and }
                                      \forall \ell' \in L, x_{\ell'} \in \itp{A_{\ell'}}.\, \itp{M_{\ell'}}(\pair{a, x_{\ell'}},b)\downarrow) \\
      \text{undefined}              & (\text{otherwise})
    \end{cases} \\
   &\quad(\text{where $\ctx \stoup \kctx \ent V \T \sumtype{\ell \in L} A_\ell$ and }
          \ctx, x_\ell \T A_\ell \stoup \kctx \ent M_\ell \T \C\text{, for } \ell \in L) \\[1ex]
     %
    %
  &\itp{\ctx \stoup \kctx \ent \force V \T \C}(c) \simeq \itp{V}(c) \\[1ex] 
  &\itp{\ctx \stoup \kctx \ent \ret V \T F A}(c) \simeq
      \eta_A(\itp{V}(c)) \\[1ex]
  &\itp{\ctx \stoup \kctx \ent \bind{M}{x \T A} N \T \C}(a,b) \\
  &\quad\simeq
    \begin{cases}
      U(\dagsub{\itp{N}(\pair{a,-},b)}{\itp{\C}})(\itp{M}(a,b)) & (\text{if }\forall x \in \itp{A}.\, \itp{N}(\pair{a,x},b)\downarrow) \\
      \text{undefined}              & (\text{otherwise})
    \end{cases} \\[1ex]
  %
  &\itp{\ctx \stoup \kctx \ent \pair{M_\ell}_{\ell \in L} \T \prodtype{\ell \in L} \C_\ell}(c) \simeq
    \pair{\itp{M_\ell}(c)}_{\ell \in L} \\[1ex]
  &\itp{\ctx \stoup \kctx \ent \prj{\ell} M \T \C_\ell}(c) \simeq
    {\Pr \ell}( \mathord{\itp{M}}(c)) \\[1ex]
  &\itp{\ctx \stoup \kctx \ent \lam{x \T A} M \T A \to \C}(a,b) \\[1ex]
  &\quad\simeq
    \begin{cases}
      \itp{M}(\pair{a, -},b) & (\text{if }\forall x \in \itp{A}. \, \itp{M}(\pair{a, x},b)\downarrow ) \\
      \text{undefined}              & (\text{otherwise})
    \end{cases} \\[1ex]
    %
  &\itp{\ctx \stoup \kctx \ent M \, V \T \C}(c) \simeq
     \itp{M}(c)(\itp{V}(c)) \\[1ex]
    %
  &\itp{\ctx \stoup \kctx \ent \op{V}(\abs{x \T \beta} M) \T \C}(a,b) \\
  &\quad\simeq
    \begin{cases}
      \op*_{\itp{\C}}(\itp{V}(a,b),
      \itp{M}(\pair{a,-},b)) & (\text{if }\forall x \in \itp{\beta}.\, \itp{M}(\pair{a, x},b)\downarrow ) \\
      \text{undefined}              & (\text{otherwise})
    \end{cases} \\[1ex]
    %
  &\itp{\ctx \stoup \kctx \ent k_j(V) \T \C_i}(a, b) \simeq  
  \Pr j(b)(\itp{V}(a, b))  
 \cond{1 \leq j \leq n}
\end{align*}
In the above, we have used \emph{existence}  assertions $e\downarrow$ which hold if, and only if, the expression~$e$ is defined. Typing judgements are again abbreviated on the right-hand side and
we  used the evident tupling and projections associated to labelled products. 

Consider a handler
$
  \ctx \stoup \kctx \ent H \T \handtype{\C}
$
where
\[H =  \{
    \op{x \T \alpha}(k \T \conttype{\beta}{\C}) \maps M_{\op*}
  \}_{\op* \T \optype{\alpha}{\beta}}
\]
The interpretations of the handling terms $M_{\op*}$ in $H$ yield partial maps
\[
  \itp{\ctx, x \T \alpha \stoup \kctx, k \T \conttype{\beta}{\C} \ent M_{\op*} \T \C}
    \from ( \itp{\ctx} \times \itp{\alpha} ) \times ( \itp{\kctx} \times U \itp{\C}^{\itp{\beta}}) \pto U \itp{\C}
\]
and by suitable rearrangement and transposition, one obtains (total) maps
\[
  \op{H} \from \itp{\ctx} \times \itp{\kctx} \to (\itp{\alpha} \times U \itp{\C}^{\itp{\beta}} \pto U \itp{\C})
\]
Hence, for any $c \in \itp{\ctx} \times \itp{\kctx}$, we obtain partial maps 
$\op{H}(c)$ on $U \itp{\C}$.

\begin{defi}
  A handler $\ctx \stoup \kctx \ent H \T \handtype{\C}$ is \emph{correct at $c$} 
  if the $\op{H}(c)$ are all total  
   and the set $U \itp{\C}$ together with the $\op{H}(c)$ forms a model of $\thy$. 
  %
  The handler is \emph{correct} if it is correct at all $c$.
\end{defi}

We can then give the interpretation of handlers:
\[
  \itp{\ctx \stoup \kctx \ent H \T \handtype{\C}}(c) \simeq
    \begin{cases}
      \langle \op{H}(c) \rangle_{\op*} & \cond{\text{if $\ctx \stoup \kctx \ent H \T \handtype{\C}$ is correct at $c$}} \\
      \text{undefined} & \cond{\text{otherwise}}
    \end{cases}
\]
and the interpretation of the handling construct follows the lines discussed informally above:
%
%
%
%
\begin{multline*}
  \itp{\ctx \stoup \kctx \ent \handleto{M}{H}{x \T A} N \T \C}(a,b) \\
  \simeq \left \{  \begin{array}{ll}
U(\dagsub{\itp{N}(\pair{a,-},b)}{
\pair{ U \itp{\C},\itp{H}(a,b)}
})(\itp{M}(a,b)) & (\forall x \in \itp{A}.\, \itp{N}(\pair{a,x},b)\downarrow\\
                                                                                                & \mbox{and } H \mbox{ is correct at } \pair{a,b}) \\
\mbox{undefined} & (\mbox{otherwise})
\end{array}
\right .
\end{multline*}
%

The existence conditions adopted above for the existence of term denotations are quite strict. For example, in the case of the first matching construct one might instead have written:
\[\itp{\ctx \stoup \kctx \ent \matchpair{V}{x, y}{M} \T \C}(a, b) \simeq
    \itp{M}(\pair{a, \Pr 1( \itp{V}(a, b)), \Pr 2 (\itp{V}(a, b))}, b)\]
and there are similar alternatives for terms of any of the forms $\matchcase{V}{\ell(x_\ell) \maps M_\ell}{\ell \in L} $ or  $ \bind{M}{x \T A} N$ or  $\handleto{M}{H}{x \T A} N$. Our choices followed the principle that the denotation of a term should exist if, and only, that of all its subterms exist (taking due account of variable binding occurrences).

The correctness of handlers is clearly of central concern.
If the effect theory is empty, then any handler is correct,
  but, in general --- and unsurprisingly --- correctness is undecidable.
Indeed, as will be discussed in Section~\ref{sec:correctness_of_handlers},
  even the decidability of the correctness of very simple handlers is a $\Pi_2$-complete problem.

\begin{rem}
\label{rem:correctness}
  All the handlers, given in Examples~\ref{exa:handlers} and Section~\ref{sec:examples} are correct.
  In particular, the exception handler and the stream redirection handlers are trivially correct as the corresponding effect theories are empty. Also, the CCS relabelling and restriction handlers are correct for both the weak and the strong bisimulation theories.

  Both the rollback handlers are also correct in the presence of destructive exceptions, given in Examples~\ref{exa:effect_theories}.
  The standard exception handler, however, is not correct for destructive exceptions,
    as it intercepts an exception and provides a replacement computation instead,
    but does not correct any modifications made to the state.
  In particular, the handler does not respect the equation
  \[
    \setop{\locvar, x}(\raiseop{e}()) = \raiseop{e}()
  \]
\end{rem}

  Another possible approach that may well be worth investigating to giving our language a denotational semantics would be to consider only terms that are guaranteed to receive an interpretation under suitable assumptions.
  One could 
  introduce conditional judgements of the form
    $\ctx \stoup \kctx \stoup \Phi \ent V \T A$
  (and similarly for computations and handlers), where $\Phi$ is a list of universally quantified equations  
guaranteeing that handlers in $V$ respect the equations of the effect theory. 
One would also have a logic for establishing the assertions (the logic for algebraic effects of~\cite{plotkin08a-logic} should be helpful here). Thus the typing judgement and the logic would be 
  integrated in a way somewhat reminiscent of 
  Martin-L\"{o}f type theory. 


\newcommand{\oldsem}[1]{}

\oldsem{\subsection{Interpreting values and computations}

We assume given
  a signature (so that value and computation terms are determined),
  an effect theory $\thy$ over that signature, and
  an interpretation of the effect theory
  (so that the category of models, etc., is determined).

We can then interpret value types $A$ by sets $\itp{A}$ and computation types $\C$ by models $\itp{\C}$ of the given effect theory~$\thy$.
Value types are interpreted in the same way as signature types,
  except for $U \C$, which is interpreted as the carrier of $\itp{\C}$.
The computation type $F A$ is interpreted by the free model on $\itp{A}$.
Products of computation types are interpreted by product models and function types are interpreted by exponent models.
These are defined as follows:
\begin{defi}
  For any family of models $\model_i$, where $i$ ranges over the index
  set~$I$, the \emph{product model}~$\prod_{i \in I} \model_i$ is
   the model with carrier $\prod_{i \in I} |\model_i|$ and
  component-wise defined operations.
  For any model $\model$ and any set $A$, the
  \emph{exponent model}~$\model^A$ is the model with
  carrier $| \model|^A$ and pointwise defined operations.
\end{defi}
Contexts are again interpreted componentwise,
  with $x_1 \T A_1, \dots, x_m \T A_m$ being interpreted by
  $\itp{A_1} \times \dots \times \itp{A_m}$,
  and
  $k_1 \T \conttype{\beta_1}{\C_1}, \dots, k_n \T \conttype{\beta_n}{\C_n}$ being interpreted by
  $U \itp{\C_1}^{\itp{\beta_1}} \times \dots \times U \itp{\C_n}^{\itp{\beta_n}}$.

We wish to interpret
  well-typed values $\ctx \stoup \kctx \ent V \T A$ by maps
  $\itp{V} \from \itp{\ctx} \times \itp{\kctx} \to \itp{A}$,
  well-typed computations $\ctx \stoup \kctx \ent M \T \C$ by maps
  $\itp{M} \from \itp{\ctx} \times \itp{\kctx} \to U \itp{\C}$,
  and well-typed handlers $\ctx \stoup \kctx \ent H \T \handtype{\C}$ by maps
  \[
    \itp{H}\from \itp{\ctx} \times \itp{\kctx}
    \to
    \prod_{\op* \T \optype{\alpha}{\beta}} \itp{\alpha}
    \times
    U \itp{\C}^{\itp{\beta}} \rightarrow U \itp{\C}
  \]
  such that, for every $c \in  \itp{\ctx} \times \itp{\kctx} $, the set $U \itp{\C}$,
  together with the operations $\Pr {\op*}(\itp{H}(c))$, forms a model of  $\thy$.

However, as discussed informally above, not all handlers  receive an interpretation.
As handler terms occur in computation terms and computation terms occur in value terms,
  computation terms and value terms may also not receive an interpretation.
Consequently we only have that,
  if the interpretation $\itp{V}$ of a well-typed value $\ctx \stoup \kctx \ent V \T A$ is defined,
  then it is a map $\itp{\ctx} \times \itp{\kctx} \rightarrow \itp{A}$,
  and similarly for computation terms and handler terms.
When giving interpretations it is then convenient to use \emph{Kleene equality $\simeq$}
  where two sides are equal if either they are both defined and equal or they are both undefined.










\begin{rem}
  \label{rem:conditional}
There are other approaches to giving our language a denotational semantics. It is possible to modify the  definition of the interpretation, so that, e.g.,  the interpretation $\itp{V}$ of a well-typed value $\ctx \stoup \kctx \ent V \T A$  can be defined for some, if not all,  values of $\itp{\ctx} \times \itp{\kctx} $; that is, $\itp{V}$ is interpreted as a partial function $\itp{\ctx} \times \itp{\kctx} \rightharpoonup  \itp{A}$. However it is not clear to us that this approach would generalise well to a suitable class of categories as the sources of partiality, a failure to obey equations, would have to fit with a notion of partial morphism.

  Another possible approach would be to consider only terms that are guaranteed to receive an interpretation.
  In this case, one could, for example, introduce conditional judgements of the form
    $\ctx \stoup \kctx \stoup \Phi \ent V \T A$
  (and similarly for computations and handlers), where $\Phi$ is a list of equations. 
  In particular, $\Phi$ would consist of equations stating that handlers in $V$ respect the equations of the effect theory $\thy$. One would also have a logic for establishing the assertions (the logic for algebraic effects should he helpful here). Thus the typing judgement and the logic would be closely integrated as in, for example, Martin-L\"{o}f type theory (although for a different reason). 

But since this results in a more involved type system,  and since the aim of this paper is to introduce the concept of handlers, we opted for the simpler approach employing Kleene equality.
\end{rem}

Fix contexts
  $\ctx = x_1 \T \alpha_1, \dots, x_m \T \alpha_m$ and
  $\kctx = k_1 \T \conttype{\beta_1}{\C_1}, \dots, k_n \T \conttype{\beta_n}{\C_n}$.
Value terms are  interpreted as follows:
\begin{align*}
  \itp{\ctx \stoup \kctx \ent x_i \T A_i} &\simeq  \Pr i \circ \Pr 1 \cond{1 \leq i \leq m}\\
  \itp{\ctx \stoup \kctx \ent \f(V) \T \alpha} &\simeq
    \itp{\f} \circ \itp{V}
    \cond{\f \T \alpha \to \beta} \\
  \itp{\ctx \stoup \kctx \ent \one \T \unit} &\simeq
    \mathrm{t} \\
  \itp{\ctx \stoup \kctx \ent \pair{V, W} \T A \times B} &\simeq
    \pair{\itp{V}, \itp{W}} \\
  \itp{\ctx \stoup \kctx \ent \ell(V) \T \sumtype{\ell \in L} A_\ell} &\simeq
    \In \ell \circ \itp{V} \\
  \itp{\ctx \stoup \kctx \ent \thunk M \T U \C} &\simeq  \itp{\ctx \stoup \kctx \ent M \T \C}
\end{align*}
In the above,
  typing judgements are abbreviated on the right-hand side,
  $\mathrm{t} \from A \to \mathbbm{1}$ is the map to the one-element set,
  and $\In \ell \from A_\ell \to \sum_{\ell \in L} A_\ell$ is the $\ell$-th injection into a disjoint sum.

Next, computation terms, other the handling construct (for which see below) are interpreted as follows:
\begin{align*}
  \itp{\ctx \stoup \kctx \ent \matchpair{V}{x, y}{M} \T \C} &\simeq
    \itp{M}
    \circ \pair{\Id[\itp{\ctx} \times \itp{\kctx}], \itp{V}} \\
  \itp{\ctx \stoup \kctx \ent \matchcase{V}{\ell(x_\ell) \maps M_\ell}{\ell \in L} \T \C} &\simeq
    [\itp{M_\ell}]_{\ell \in L}
    \circ \Ds
    \circ \pair{\Id[\itp{\ctx} \times \itp{\kctx}], \itp{V}} \\
  \itp{\ctx \stoup \kctx \ent \force V \T \C} &\simeq  \itp{\ctx \stoup \kctx \ent V \T U \C} \\
  \itp{\ctx \stoup \kctx \ent \ret V \T F A} &\simeq
    \eta_A \circ \itp{V} \\
  \itp{\ctx \stoup \kctx \ent \bind{M}{x \T A} N \T \C} &\simeq
    {\itp{N}^{\dagger_{\itp{\C}}} }
      \circ \pair{\Id[\itp{\ctx} \times \itp{\kctx}],\itp{M}} \\
  \itp{\ctx \stoup \kctx \ent \pair{M_\ell}_{\ell \in L} \T \prodtype{\ell \in L} \C_\ell} &\simeq
    \pair{\itp{M_\ell}}_{\ell \in L} \\
  \itp{\ctx \stoup \kctx \ent \prj{\ell} M \T \C_\ell} &\simeq
    {\Pr \ell} \circ \mathord{\itp{M}} \\
  \itp{\ctx \stoup \kctx \ent \lam{x \T A} M \T A \to \C} &\simeq
    \Tr{\itp{\ctx, x \T A \stoup \kctx \ent M \T \C}} \\
  \itp{\ctx \stoup \kctx \ent M \, V \T \C} &\simeq
    \Ev \circ \pair{
      \itp{M},
      \itp{V}
    } \\
  \itp{\ctx \stoup \kctx \ent \op{V}(\abs{x \T \beta} M) \T \C} &\simeq
    \op*_{\itp{\C}} \circ \pair{
      \itp{V},
      \Tr{\itp{\ctx, x \T \beta \stoup \kctx \ent M \T \C}}
    } \\
  \itp{\ctx \stoup \kctx \ent k_j(V) \T \C_i} &\simeq  \Ev \circ \pair{\Pr j \circ \Pr 2, \itp{V}} \cond{1 \leq j \leq n}
\end{align*}

In the above, typing judgements are again abbreviated on the right-hand side, \emph{parameterised lifting}
    $f^{\dagger_{\model}} \from A \times U F B \to U \model$
of a map $f \from A \times B \to U \model$ is given by
\[
  {f}^{\dagger_{\model}} \defeq U \dag{f}{\model} \circ \St{A, B} \from A \times U F B \to U \model
\]
where $\bar{f}$ is the homomorphism induced by $f$,
and $\St{A,B} \T A \times U F B \rightarrow U F (A \times B)$ is the \emph{strength} of the monad $U F$;
we have also used the evident tupling and projections associated to labelled products.

Consider a handler
$
  \ctx \stoup \kctx \ent H \T \handtype{\C}
$
where
\[H =  \{
    \op{x \T \alpha}(k \T \conttype{\beta}{\C}) \maps M_{\op*}
  \}_{\op* \T \optype{\alpha}{\beta}}
\]
If defined, the interpretations of the handling terms $M_{\op*}$ in $H$ yield maps
\[
  \itp{\ctx, x \T \alpha \stoup \kctx, k \T \conttype{\beta}{\C} \ent M_{\op*} \T \C}
    \from ( \itp{\ctx} \times \itp{\alpha} ) \times ( \itp{\kctx} \times U \itp{\C}^{\itp{\beta}}) \to U \itp{\C}
\]
and by suitable rearrangement and transposition, one obtains maps
\[
  \op{H} \from \itp{\ctx} \times \itp{\kctx} \to (\itp{\alpha} \times U \itp{\C}^{\itp{\beta}} \to U \itp{\C})
\]
Hence, for any $c \in \itp{\ctx} \times \itp{\kctx}$, we obtain operations $\op{H}(c)$ on $U \itp{\C}$.

\begin{defi}
  A handler $\ctx \stoup \kctx \ent H \T \handtype{\C}$ is \emph{correct} if the $\op{H}$ are all defined and if, for all $c \in \itp{\ctx} \times \itp{\kctx}$, the set $U \itp{\C}$, together with the $\op{H}(c)$, is a model of $\thy$.
\end{defi}

We can then give the interpretation of handlers:
\[
  \itp{\ctx \stoup \kctx \ent H \T \handtype{\C}} \simeq
    \begin{cases}
      c \mapsto \langle \op{H}(c) \rangle_{\op*} & \cond{\text{if $\ctx \stoup \kctx \ent H \T \handtype{\C}$ is correct}} \\
      \text{undefined} & \cond{\text{otherwise}}
    \end{cases}
\]

\newcommand{\Trm}[1]{\Map{tr}^{-1}(#1)}
The interpretation of the handling construct
\[\ctx \stoup \kctx \ent \handleto{M}{H}{x \T A} N \T \C\]
is a bit complicated because of the possibility of free variables in $H$.
Suppose that the interpretations of $M$, $H$, and $\C$ are all defined (necessary for that of the handling construct to be defined).
Then for every
  $c \in \itp{\ctx} \times \itp{\kctx}$
define $\model_c$ to be the model of $\thy$ with carrier $U \itp{\C}$ and operations $\Pr {\op*}(\itp{H}(c))$.
Taking the product of all these models we obtain a model $\model_H$ of $\thy$ with carrier
  $U \itp{\C}^{\itp{\ctx} \times \itp{\kctx}}$.
Next, transposing $\itp{\ctx, x \T A \stoup \kctx \ent N \T \C}$ (modulo a rearrangement of factors)
  we obtain a map  $\itp{A} \to U \itp{\C}^{\itp{\ctx} \times \itp{\kctx}}$.
The adjoint homomorphism of this map has type $F \itp{A} \lol \itp{\C}^{\itp{\ctx} \times \itp{\kctx}}$ and this,
  in turn, yields a map
 $\itp{\ctx} \times \itp{\kctx} \times U F \itp{A} \to U \itp{\C}$.
Combining this last map with the interpretation of $M$ we finally obtain the desired map $\itp{\ctx} \times \itp{\kctx} \to U \itp{\C}$.

Writing this out we obtain the interpretation of the handling construct:
\begin{multline*}
  \itp{\ctx \stoup \kctx \ent \handleto{M}{H}{x \T A} N \T \C} \simeq \\
  \Trm{\Tr{\itp{\ctx, x \T A \stoup \kctx \ent N \T \C}}^{\dagger_{\model_H}} }
  \circ
  \pair{
    \Id[\itp{\ctx} \times \itp{\kctx}],
    \itp{\ctx \stoup \kctx \ent M \T FA}
  }
\end{multline*}

The correctness of handlers is clearly of central concern.
If the effect theory is empty, then any handler is correct,
  but, in general --- and unsurprisingly --- correctness is undecidable.
Indeed, as will be discussed in Section~\ref{sec:correctness_of_handlers},
  even the decidability of the correctness of very simple handlers is a $\Pi_2$-complete problem.

\begin{rem}
\label{rem:correctness}
  All the handlers, given in Examples~\ref{exa:handlers} and Section~\ref{sec:examples} are correct.
  In particular, the exception handler and the stream redirection handlers are trivially correct as the corresponding effect theories are empty. Also, the CCS relabelling and restriction handlers are correct for both the weak and the strong bisimulation theories.

  Both the rollback handlers are also correct in the presence of destructive exceptions, given in Examples~\ref{exa:effect_theories}.
  The standard exception handler, however, is not correct for destructive exceptions,
    as it intercepts an exception and provides a replacement computation instead,
    but does not correct any modifications made to the state.
  In particular, the handler does not respect the equation
  \[
    \setop{\locvar, x}(\raiseop{e}()) = \raiseop{e}()
  \]
\end{rem}


}

\section{Reasoning about handlers}
\label{sec:reasoning_about_handlers}
\newcommand{\omitcontext}[1]{ }
\newcommand{\defined}{\!\downarrow}






The semantics we have just given allows us to reason about handlers.
We are interested in two questions in particular:
  which computations are equal,
  and which handlers are correct?
We content ourselves in this section with  some initial observations, and do not attempt to provide  a 
full-fledged 
logic for algebraic effects and handlers.

\newcommand{\myeq}{\bumpeq}
We fix a signature, and an effect theory $\thy$ and its interpretation, and consider as well-formed formulas  $\ctx \stoup \kctx \ent \varphi$  those that can be built up from atomic formulas by the usual boolean connectives and universal and existential quantification  over value types, e.g., $\forall x\T A.\, \varphi$, and continuations, e.g., $\forall  k \T \conttype{\alpha}{\C}.\, \varphi $. As atomic formulas we take existence (or definedness) and Kleene equality assertions for value terms, $V\!\defined$ and $V \simeq W$ (with evident well-formedness conditions), and for computation and handler terms.  In the case of a handler $H$ the existence assertion $H\defined$ can equivalently be understood as that of  the correctness of  $H$.  There is a convenient ``Kleene inequation" $V \lesssim W$, which abbreviates $V\!\defined \;\Rightarrow\; V \simeq W$ (and similarly for computation and handler terms).

There is an evident inductively defined satisfaction relation $c \models_{\ctx \stoup \kctx } \varphi$ for $c \in \itp{\ctx} \times \itp{\kctx}$ and $ \varphi$  well-formed relative to $\ctx$ and $\kctx$. In particular $c \models_{\ctx \stoup \kctx } V\defined$ holds if, and only if, $\itp{V}(c)\downarrow$ and $c \models_{\ctx \stoup \kctx } V \simeq W$ holds if, and only if,  $\itp{V}(c) \simeq \itp{W}(c)$, and similarly for computation and handler terms. We say that $\models_{\ctx \stoup \kctx } \varphi$ holds if, and only if $c \models_{\ctx \stoup \kctx } \varphi$ holds for all $c \in \itp{\ctx} \times \itp{\kctx}$; below we generally just say that $\varphi$ holds, and the intended contexts $\ctx$ and $\kctx$ will be clear. 

%
%
%
%



We first consider equations. 
Versions of Levy's call-by-push-value equations~\cite{levy06call-by-push-value}   hold.
For example we have a Kleene $\beta$-inequality for the sequencing construct:
\[
  \omitcontext{\ctx \stoup \kctx \ent} \bind{\ret x}{x \T A} M \lesssim M 
\]
and a Kleene $\eta$-equality for functions:
\[
  \omitcontext{\ctx \stoup \kctx \ent} \lam{x \T A} M x \simeq M
\]
Next, we have equations that describe the component-wise and pointwise behaviour of operations on product and function types, respectively:
\begin{align*}
  \omitcontext{\ctx \stoup \kctx \ent} \op{V}(\abs{x \T \alpha} \pair{M_\ell}_{\ell \in L}) &\simeq \pair{\op{V}(\abs{x \T \alpha} M_\ell)}_{\ell \in L} \T \prodtype{\ell \in L} \C_\ell \\
  \omitcontext{\ctx \stoup \kctx \ent} \op{V}(\abs{x \T \alpha} \lam{y \T A} M) &\simeq \lam{y \T A} \op{V}(\abs{x \T \alpha} M) \T A \to \C
\end{align*}
We also have an equation giving the commutativity between operations and sequencing:
\[
  \omitcontext{\ctx \stoup \kctx \ent} \bind{\op{V}(\abs{x \T \alpha} M)}{y \T A} N \simeq \op{V}(\abs{x \T \alpha} \bind{M}{y \T A} N) \T \C
\]
This equation holds as sequencing is defined using the homomorphism induced by the universality of the free model,
  thus it maps an operation on $F A$ to an operation on $\C$.
The equations given by Levy~\cite{levy06call-by-push-value} for two specific operations: printing and divergence (divergence is considered in Section~\ref{sec:recursion}), are  instances of the above three equations.

Further equations  are inherited from 
from the effect theory $\thy$.
%
Given a  a template
  \[
    x_1 \T \alpha_1, \dots, x_m \T \alpha_m \stoup z_1 \T \beta_1, \dots, z_n \T \beta_n \ent T
  \]
  and 
 distinct  continuation variables $k_1, \dots, k_n$ we write $T[k_1/z_1,\ldots, k_n/z_n]$ for the computation term obtained by replacing each occurrence of a $z_i$ in $T$ by one of the corresponding $k_i$. If $\ctx$ contains the type assignments $x_1 \T \alpha_1, \dots, x_m \T \alpha_m$ and $\kctx$ contains  the type assignments $k_1 \T \conttype{\beta_1}{\C} , \dots, k_n \T \conttype{\beta_n}{\C}$ then $\ctx \stoup \kctx \ent T[k_1/z_1,\ldots, k_n/z_n] \T \C$ holds.
%
%
We then have that  \[T_1[k_1/z_1,\ldots, k_n/z_n] \ \simeq T_2[k_1/z_1,\ldots, k_n/z_n] \]
holds for every equation 
\[x_1 \T \alpha_1, \dots, x_m \T \alpha_m \stoup z_1 \T \beta_1, \dots, z_n \T \beta_n \ent T_1 = T_2 \]
in $\thy$ and for every type assignment $k_1 \T \conttype{\beta_1}{\C} , \dots, k_n \T \conttype{\beta_n}{\C}$.
%


We next give two Kleene inequations for the handling construct that state the universal properties of the induced homomorphism:
 it extends the inducing map on values, and it acts homomorphically on operations.
For any handler
  $H = \{ \op{y}(k) \maps M_{\op*} \}_{\op* \T \optype{\alpha}{\beta}}$,
the inequations are:
\begin{align*}
  \omitcontext{\ctx \stoup \kctx \ent} \handleto{\ret x}{H}{x \T A} N &\;\lesssim\; N \\
  \omitcontext{\ctx \stoup \kctx \ent} \handleto{\op{y}(\abs{x' \T \beta} M)}{H}{x \T A} N &\;\lesssim\;
  M_{\op*}[\abs{x' \T \beta}  \handleto{M}{H}{x \T A} N/k] 
\end{align*}
%
where substitutions of the form $\abs{x' \T \beta} M' / k$ replace each occurrence of a computation $k(V)$ by $M'[V / x']$ ---
recall that continuation variables always appear only in the form~$k(V)$.

As mentioned in Section~\ref{terms}, we have:
\[
  \omitcontext{\ctx \stoup \kctx \ent} \bind{M}{x \T A} N \simeq \handleto{M}{\{ \}}{x \T A} N
\]
stating that sequencing is equivalent to the special case of handling in which all operations are handled by themselves.
One can then observe that (given the existence assertions below) the two handler inequations generalise the formulas given above for the sequencing construct, viz. the $\beta$-inequality and the commutativity of operations and sequencing.

There are other equations for exception handlers,
  given by Benton and Kennedy~\cite{benton01exceptional}, and Levy~\cite{levy06monads},
  and one would wish to generalise these to our general handlers.
It turns out that these equations fail for general handlers,
  but do hold for particular classes of handlers, see:~\cite{pretnar10the-logic}.
We do not consider this issue further here.

We next consider existence assertions. They all spell out the conditions for term interpretations to exist. For example the following hold:
%
\begin{align*}
\ret V \defined \,\;&\Leftrightarrow\; \, V\defined \\
\bind{M}{x \T A} N\defined \,\;&\Leftrightarrow\; \, M\defined \,\wedge\, \forall x \T A.\, N \defined \\
\lam{x \T A} M \defined \,\;&\Leftrightarrow\; \, \forall x \T A.\, M \defined \\
\op{V}(\abs{x \T \alpha}  M)\defined \,\;&\Leftrightarrow\; \, V\defined \,\wedge\, \forall x \T \alpha.\, M \defined
\end{align*}

Turning to handler existence, or correctness, we first need to be able to replace operations in templates by their corresponding definitions in  handlers.
So, let us consider a handler $\ctx \stoup \kctx \ent H \T \handtype{\C}$,  where $H$ is
    $\{ \op{x \T \alpha}(k \T \conttype{\beta}{\C}) \maps M_{\op*} \}_{\op* \T \optype{\alpha}{\beta}}$,
and a template variable context $\tctx$.
Let $\kctx'$  be the continuation context with a continuation variable $k_z \T \conttype{\alpha}{\C}$ for each template variable $z \T \alpha \in \tctx$.
Then, for every template term $\ctx' \stoup \tctx \ent T$
  we recursively define a computation term  $\ctx, \ctx' \stoup \kctx, \kctx' \ent T^H \T \C$ by:
\begin{align*}
  z(V)^H &= k_z(V) \\
  (\matchpair{V}{x_1, x_2}{T})^H &= \matchpair{V}{x_1, x_2}{T^H}\\
  (\matchcase{V}{\ell(x_\ell) \maps T_\ell}{\ell \in L})^H &=
    \matchcase{V}{\ell(x_\ell) \maps T_\ell^H}{\ell \in L} \\
  \op{V}(\abs{y \T \beta} T)^H &=
    M_{\op*}[V / x, \abs{y \T \beta} T^H / k]
\end{align*}
The following then holds:
\[ H \defined \;\;\, \Leftrightarrow \;\;\,  \bigwedge \{  \forall x \T \alpha.\, M_{\op*} \defined \; \mid\; \op* \T \optype{\alpha}{\beta} \}
\,\wedge\,
\bigwedge \{ T_1^H \simeq T_2^H \;\mid\; \ctx' \stoup \tctx \ent T_1 = T_2 \in \thy\}
\]
%
%
asserting that a handler is correct when its operations are  defined and respect the equations of the effect theory. In particular, $\{\}\defined$ holds, where $\{\}$ is the empty handler, the one defining all operations by themselves.
Regarding the handling construct itself, the following holds:
\[\handleto{M}{H}{x \T A} N \defined \,\;\Leftrightarrow\; \,  M\defined \,\wedge\, H\defined  \,\wedge\,  \forall x \T A.\, N\defined  \]

%
%

\section{Deciding handler correctness}
\label{sec:correctness_of_handlers}

The decidability of handler correctness is an interesting question and one potentially pertinent for compiler writers.
We now give some results on the decidability of some natural classes of handlers;
their proofs are given in Appendix~\ref{apdx:decidability_of_correctness}.

\begin{defi}
  A handler $\ctx \stoup \kctx \ent H \T \handtype{\C}$ over a given signature
  is \emph{simple} if 
  \begin{itemize}
  \item up to reordering, $\ctx$ has the form
  \[
    x_1 \T \alpha_1, \dots, x_m \T \alpha_m, f_1 \T U (\beta_1 \to \C), \dots, f_n \T U (\beta_n \to \C)
  \]
 \item  up to reordering, $\kctx$ has the form
  \[
    k'_1 \T \conttype{\beta'_1}{\C}, \ldots, k'_p \T \conttype{\beta'_p}{\C}
  \]
  and
  \item for each $\op* \T \optype{\alpha}{\beta}$, there is a template
  \[
   x_1 \T \alpha_1, \dots, x_m \T \alpha_m, x \T \alpha
   \stoup z_1 \T \beta_1, \dots, z_n \T \beta_n,
   z'_1 \T \beta'_1, \dots, z'_p \T \beta'_p,
   z \T \beta \ent T_{\op* }
  \]
  such that  the handling term
  \[
    \ctx, x \T \alpha \stoup \kctx, k \T \conttype{\beta}{\C} \ent M_{\op*} \T \C
  \]
  is obtained by the substitution of $T_{\op*}$
  that replaces each $x_i$ by itself, $x$ by itself,
  each $z_j$ by $(y_j\T \beta_j.\, (\force f_j) y_j)$, 
  each $z'_l$ by  $(y'_l \T \beta'_l.\, k'_l(y'_l))$,   
  and $z$ by $(y \T \beta.\, k(y))$.
  %

  \end{itemize} 
\end{defi}

\noindent In essence, simple handlers define handling computations in terms of 
algebraic constructs only.
The temporary-state handler, the CCS relabelling and restriction handlers, and stream redirection handlers are all simple;
none of the parameter-passing handlers are simple as they all contain lambda abstractions in their handling terms.
The exception handler
\[
  \{
    \raiseop{y \T \exc}(k \T \conttype{\void}{\C}) \maps \matchcase{y}{e(z) \maps N_e}{e \in \exc}
  \}
\]
is also not simple, as computation terms $N_e$ may be arbitrary.
One can instead use the 
simple handler
\[
  f \T U (\exc \to \C) \ent \{
    \raiseop{e \T \exc}(k \T \conttype{\void}{\C}) \maps (\force f) \, e
  \}
\]
Then, letting $f$ be the thunk of
\[
  \lam{y \T \exc} \matchcase{y}{e(z) \maps N_e}{e \in \exc}
\] 
one obtains exactly the same behaviour.

A signature is \emph{simple} if it has no base types or function symbols.
In that case there is a unique interpretation, the trivial one, and we will omit mention of it.
Simple signatures and theories are equivalent to ones in which all operation symbols are $\type{n}$-ary for some $n$.
The signatures given above for exceptions (over a finite set), read-only state, I/O, and nondeterminism are all simple, but none of the others are.

\begin{thm} \label{thm:pi2}
 The problem of deciding, given a simple signature, effect theory, and closed simple handler
 $ \ent H \T \handtype{F \void\,}$ over the given signature,
 whether the handler is correct, is $\Pi_2$-complete.
\end{thm}

The polymorphic nature of template variables means that a template can define a whole family of handlers, one for each computation type.
This leads one to the definition of a uniformly simple family of handlers:

\begin{defi}
  A family of handlers $\{\ctx_{\C} \stoup \kctx_{\C} \ent H_{\C} \T \handtype{\C}\}_{\C}$ over a given signature,
  where $\C$ ranges over all computation types,
  is \emph{uniformly simple} if 
  \begin{itemize}
  \item up to reordering, $\ctx_{\C}$ is
  \[
    x_1 \T \alpha_1, \dots, x_m \T \alpha_m, f_1 \T U (\beta_1 \to \C), \dots, f_n \T U (\beta_n \to \C)
  \]
  \item up to reordering, $\kctx_{\C}$ is
  \[
    k'_1 \T \conttype{\beta'_1}{\C}, \ldots, k'_p \T \conttype{\beta'_p}{\C}
  \]
  and
  \item  for each $\op* \T \optype{\alpha}{\beta}$, there is a template
  \[
   x_1 \T \alpha_1, \dots, x_m \T \alpha_m, x \T \alpha
   \stoup z_1 \T \beta_1, \dots, z_n \T \beta_n,
   z'_1 \T \beta'_1, \dots, z'_p \T \beta'_p,
   z \T \beta \ent T_{\op* }
  \]
  such that for each computation type $\C$,
     the handling term
  \[
    \ctx_{\C}, x \T \alpha \stoup \kctx_{\C}, k \T \conttype{\beta}{\C} \ent M_{\op*} \T \C
  \]
   is obtained by the substitution of $T_{\op*}$
  that replaces each $x_i$ by itself, $x$ by itself,
  each $z_j$ by $(y_j\T \beta_j.\, (\force f_j) y_j)$, 
  each $z'_l$ by  $(y'_l \T \beta'_l.\, k'_l(y'_l))$,   
  and $z$ by $(y \T \beta.\, k(y))$.
  %

  \end{itemize}
  
\end{defi}

\noindent If a family of handlers is uniformly simple then, as is evident, any member of the family is simple.
Conversely, because of the  polymorphic nature of templates, any simple handler can be generalised to obtain a uniformly simple one.
Of the simple handlers given above, the temporary-state handler, the stream redirection handlers, and the alternative exception handler are also uniformly simple; the CCS relabelling and restriction handlers are not as we define them only for the computation type $F \void$.

Note that a uniformly simple family of handlers is determined by the various variables, types and template terms discussed in the above definition.
As there are finitely many of these altogether, uniformly simple families of handlers can be finitely presented and so it is proper to ask decidability questions about them.

\begin{defi}
  A family of handlers $\{\ctx_{\C} \stoup \kctx_{\C} \ent H_{\C} \T \handtype{\C}\}_{\C}$ over a given signature,
  where $\C$ ranges over all computation types,
  is \emph{correct} if each handler in the family is correct.
\end{defi}

Since a uniformly simple family of handlers cannot use properties of a specific computation type,
  it cannot be as contrived as an arbitrary family of handlers;
we may therefore expect it to be easier to decide its correctness.
As we now see, correctness can become semidecidable:
\begin{thm}
  \label{thm:sigma1}
  The problem of deciding, given a simple signature, effect theory,
  and a uniformly simple family of closed handlers $\{ \ent H_{\C} \T \handtype{\C\,}\}_{\C}$ over the given signature,
  whether the family of handlers is correct, is $\Sigma_1$-complete.
\end{thm}

Effect theories with a simple signature correspond
  (see the proof of Proposition~\ref{prop:free})
  to ordinary finite equational theories,
  i.e., those with finitely many finitary function symbols and finitely many axioms.
Consequently notions such as decidability can be transferred to them.
(All the above examples of simple signatures and theories are decidable in this sense.)
With this understanding, we have:
\begin{thm} \label{thm:decidable}
 The problem of deciding, given a simple signature, decidable effect theory,
 and a uniformly simple family of 
 closed handlers $\{ \ent H_{\C} \T \handtype{\C\,}\}_{\C}$ over the given signature,
 whether the family of handlers is correct, is decidable.
\end{thm}

\section{Recursion}
\label{sec:recursion}



We now sketch some changes to the syntax and semantics that enable the incorporation of recursion. We begin with syntax. Signatures are as before except that we always assume an additional operation symbol $\botop* \T \optype{\unit}{\void}$ representing nontermination.
To provide a recursion facility, we follow~\cite{levy06call-by-push-value} and
  add a fixed-point constructor $\rec{x \T U \C} M$, typed as follows:
\[
  \deduct{
    \ctx, x \T U \C \stoup \kctx \ent M \T \C
  }{
    \ctx \stoup \kctx \ent \rec{x \T U \C} M \T \C
  }
\]

\newcommand{\cpo}{$\omega$-cpo }
\newcommand{\cpos}{$\omega$-cpos }

Turning to semantics, templates are defined as before, but effect theories now consist of
\emph{inequations} $\ctx \stoup \tctx \ent T_1 \leq T_2$ rather than equations.
We  always assume that the effect theory contains the inequation $\cdot \stoup z \T \unit \ent \botop*() \leq z$,
  stating that $\botop*$ is the least element.
  As examples one can take those of Section~\ref{sec:semantics} with this additional inequation, reading equations as conjunctions of two inequations;  see~\cite{hyland06combining} for further examples.

In order to 
interpret recursion, we move, as remarked above, from the category of sets to $\wCpo$, the category of  \cpos and continuous functions.
First we assume given an assignment of \cpos $\itp{\mathbf{b}}$ to base types $\mathbf{b}$, such that $\itp{\mathbf{b}}$ is flat and countable  in the case that $\mathbf{b}$ is an arity base type;
this yields interpretations of signature types~$\alpha$ as $\omega$-cpos~$\itp{\alpha}$, which are flat and countable in the case of arity signature types.
(A \cpo is flat if no two distinct elements are comparable.)
An interpretation of an effect theory then consists of such an assignment and
  an assignment $\itp{\f} \from \itp{\alpha} \to \itp{\beta}$ of continuous functions to function symbols 
  $\f \T \alpha \to \beta$,
  and well-formed base values $\ctx \ent V \T \alpha$ are interpreted
  as continuous functions $\itp{V} \from \itp{\ctx} \to \itp{\alpha}$ much as before.

Similarly, given an \cpo $|\model|$ together with operations, that is, continuous functions
  $\op{\model} \from \itp{\alpha} \times |\model|^{\itp{\beta}} \to |\model|$,
  for each operation symbol $\op* \T \optype{\alpha}{\beta}$,
  well-formed templates
    $\ctx \stoup \tctx \ent T$
  can be interpreted as continuous functions
    $\itp{T} \from \itp{\ctx} \times \itp{\tctx} \to |\model|$.

With that,
  one can define models  $\model$ of effect theories $\thy$
  as \cpos together with operations for each operation symbol
  such that such that $\itp{T_1} \leq \itp{T_2}$ holds
  for all the inequations $\ctx \stoup \tctx \ent T_1 \leq T_2$ in $\thy$.
Defining homomorphisms in the evident way,
  one obtains a category of models $\Mod[\thy]$,
  equipped with a forgetful functor $U \from \Mod[\thy] \to \wCpo$.
As we discuss below, the forgetful functor has a left adjoint $F \from \wCpo \to \Mod[\thy]$.
Both functors have continuous strengths and are locally continuous,
  by which is meant that they act continuously on the hom $\omega$-cpos.
It is also important to note that the unique continuous homomorphism
  $\dagsub{f}{\model} \T F A \lol \model$
  induced by a continuous map $f \T A \to U \model$ is itself a continuous function of $f$.
As effect theories assume  least elements,
  carriers of models are pointed $\omega$-cpos (i.e., $\omega$-cpos with a least element)
  and homomorphisms are strict continuous functions (i.e., continuous functions preserving least elements).

The monads $UF$ obtained in this way are the standard ones that occur in semantics.
For example, the monad obtained from the theory for exceptions together with a least element is $P \mapsto (P + \exc)_\bot$,
  where the \emph{lifting} $Q_\bot$ is the $\omega$-cpo obtained from $Q$ by adding a new least element;
the monad for nondeterminism, together with a least element, is $P \mapsto \mathcal{P}(P_\bot)$,
  where $\mathcal{P}$ is the convex powerdomain monad;
and the monad for state, together with a least element and inequations sating that the state operations are strict, is $P \mapsto (S \times P)_\bot^S$.
See~\cite{hyland06combining} for further discussion and references.

To show the existence of the free model functor and its properties, one can proceed along related lines to before.
One again has families $\op*_a$ of operations of countable arity, but now 
parameterised by elements $a$ of given \cpos $P$, and one considers inequations instead of equations between the resulting infinitary terms.
Operations are interpreted by continuous functions varying continuously over the parameter $\omega$-cpos;
this corresponds to the fact that one is working with continuous functions
  $\op{\model} \from \itp{\alpha} \times |\model|^{\itp{\beta}} \to |\model|$
with $\itp{\alpha}$ an $\omega$-cpo.





One thereby obtains continuously parameterised countably infinitary inequational theories;
these can be shown to be equivalent to the discrete countable Lawvere $\wCpo$-theories:
  see~\cite{power06hylandcountable},
  and see too \cite{plotkin06some} for further discussion of parameterised equational logic.
The existence of the free model functor and its continuity properties follow from the fact that
  $\wCpo$ is locally countably presentable as a cartesian closed category (see \cite{power06hylandcountable}).

Value types are now interpreted by $\omega$-cpos and computation types are interpreted by models whose carriers are pointed $\omega$-cpos,
  again making use of product and exponentiation models.
Values, computations (including the handling construct) and handlers are interpreted by partial maps analogously to before, except that as well as requiring functions be total, one requires that they are also continuous.

The fixed-point constructor is interpreted using the usual least fixed-point interpretation:
\begin{multline*}
    \itp{\ctx \stoup \kctx \ent \rec{x \T U \C} M \T \C}(a,b) \\
 \simeq\left \{\begin{array}{ll}
          \mu \abs{x\T \itp{U\C}} \itp{M}(\pair{a, x}, b) & (\forall x \in \itp{U \C}.\, \itp{M}(\pair{a,x},b)\downarrow)\\
          \mbox{undefined} & (\mbox{otherwise})
  \end{array}
  \right  .
\end{multline*}
%
%
where, as usual, $\mu \abs{x\T P} f(x)$ is the least fixed-point of a continuous function $f$ on a pointed \cpo $P$.

One can show  by structural induction that the sets of elements of $\itp{\ctx} \times \itp{\kctx}$ at which the denotations $\itp{\ctx \stoup \kctx \ent V \T A}$  of value terms exist form a sub-cpo of $\itp{\ctx} \times \itp{\kctx}$ and  that the denotations are continuous when restricted to that sub-cpo, and, further, that the same holds for computation terms and handlers.   (One may then observe that, as a consequence, the above continuity requirement  is redundant.)

%
Correct handlers cannot redefine $\botop*$ as the theory of $\botop*$ fixes it uniquely.
The handlers given in the various examples above, as detailed in Remark~\ref{rem:correctness},  remain correct in the presence of recursion, understanding  them as defining $\botop*$  by itself (and with the addition of the inequation stating that $\botop*$ is the least element to the corresponding theories).

\section*{Conclusion} 

The current work opens some immediate questions.
The most important is how to simultaneously handle two computations to describe parallel combinators,
  e.g., that of CCS or the \textsc{Unix} pipe combinator.
Understanding this would bring parallelism within the ambit of the algebraic theory of effects.

Next, the logical ideas of Section~\ref{sec:reasoning_about_handlers} should be worked out more fully and merged with the general logic for algebraic effects~\cite{plotkin08a-logic,pretnar10the-logic}.
There is a close correspondence between the handling construct and the free model principle in the logic, which should be examined in detail.

It would be worthwhile to extend the results of Section~\ref{sec:correctness_of_handlers} further,
  whether to wider classes of signatures, handlers, or theories.
One would also like to have analogous results for (the interpretation over) \cpos\hspace{-3pt}.

Only correct handlers can be interpreted, but, as we have seen, obtaining mechanisms that ensure correctness is hard.
One option would be to drop equations altogether, when handlers are  interpreted as models of absolutely free theories (i.e., those with no equations). This would be correctly implementable if one removed the connection with the equalities 
expected for the real effects between handled computations.




More routinely, perhaps, the work done on combinations of effects in~\cite{hyland06combining} should be extended to combinations of handlers,
  and there should be a general operational semantics~\cite{plotkin01adequacy} which includes that of Benton and Kennedy~\cite{benton01exceptional}.

In so far as possible, one would like to work in a general categorical setting as regards both the denotational semantics and the logic. Considering only the denotational semantics,  it may be possible to generalise the above semantics to work over any category $\mathbf{V}$ that is locally countably presentable  as a cartesian closed category.
Presumably one would use a more abstract notion of the effect theories of
Section~\ref{sec:semantics}
above,  
based on Lawvere $\mathbf{V}$-theories
(see \cite{plotkin04computational}) or, perhaps, just discrete Lawvere $\mathbf{V}$-theories~\cite{power06hylandcountable}. As one needs  to deal with partiality when handlers are incorrect, one would  require a suitable factorisation system.

Finally, one should develop the programming language aspects further.
For example, while call-by-push-value serves well as a fundamental calculus and as an intermediate language,
  it may be more realistic, or at least more in accordance with current practice,
  to find a formulation of handlers in a call-by-value context (and call-by-name also has some interest).
It would also be important to increase ease of programming,
  for example by allowing abstraction on handlers, rather than, as above, making use of their global variables;
  one would also like syntactic support for parametric handlers;
  some ideas along these lines can be found in~\cite{plotkin09handlers}.
In general, perhaps handlers could become more first-class entities.





\section*{Acknowledgments}

We  thank
  Andrej Bauer,
  Andrzej Filinski,
  Ohad Kammar,
  Paul Levy,
  John Power,
  Mojca Pretnar,
  Alex Simpson,
  and an anonymous referee
for their insightful comments and support.

\bibliography{bibliography}
\bibliographystyle{plain}

\appendix

\section{Decidability of handler correctness}
\label{apdx:decidability_of_correctness}






We prove our results on handler correctness by reducing correctness to related questions in equational logic.
An interpretation of an equational theory $\thy$ in another $\thy'$ is given by
  an assignment of a $\thy'$-term with free variables included in $x_1, \dots, x_n$ to every function symbol of arity $n$ of $\thy$.
This results in an interpretation of every $\thy$-term by a $\thy'$-term,
  and so in an interpretation of every $\thy$-equation by a $\thy'$-equation.

\begin{lem}
  \label{lem:pi2}
 Given a finitary equational theory with finite signature and finitely many axioms,
 and an interpretation of this theory in itself,
 it is a $\Pi_2$-complete problem to decide  whether the interpretation  holds in the initial model of the theory.
\end{lem}
\begin{proof}
  \newcommand{\und}[1]{\underline{#1}}
  The problem is clearly in $\Pi_2$ as an axiom holds in the initial model of such an equational theory
  if, and only if,
  there exists a proof of all its closed instances.
  Conversely, take a $\Pi_2$ sentence of Peano Arithmetic.
  Without loss of generality, we can assume this to be of the form
    $\fra{x} \exs{y} \varphi(x, y)$,
  where $\varphi(x,y)$ defines a primitive recursive relation $R(m,n)$.
  Changing to  $\fra{x} \exs{y} \exs{y' \leq y} \varphi(x, y')$ if necessary,
    we can assume that if $R(m,n)$ holds and $n \leq n'$, then $R(m,n')$ holds too.
  We now define  a finitary equational theory $\thy$ with finite signature and finitely many axioms,
    and an interpretation of it in itself,
    such that $\fra{x} \exs{y} \varphi(x, y)$ is  true
    if, and only if,
    the interpretation of the axioms holds in the initial model of the theory.

  Let $f_1, \dots, f_k$ be a sequence of primitive recursive functions,
    including a function (coding) disjunction,
    each definable in terms of the previous ones by composition or primitive recursion,
    such that $f_k$ is the characteristic function of $R$.
  Take a constant symbol $\f[zero]$ and a unary function symbol $\f[succ]$.
  Next, for all the primitive recursive functions $f_i$ take:
    a function symbol $\f_i$ of the same arity, and
    axioms corresponding to the primitive recursive definition of $f_i$,
    written using $\f[zero]$, $\f[succ]$, and the $\f_j$ with $j < i$.

  Next, take a binary function symbol $\f[try]$ and a unary function symbol $\f[exists]$,
    together with the following two axioms:
  \begin{align*}
    \f[try](x, y) &= \f_k(x, y)  \lor \f[try](x, \f[succ](y)) \\
    \f[exists](x) &= \f[try](x, \und{0})
  \end{align*}
  where $\lor$ is the function symbol corresponding to disjunction and $\und{n}$ is the $n$-th numeral, defined using $\f[zero]$ and $\f[succ]$. All this defines the theory $\thy$.

  Finally, take the interpretation of the theory $\thy$ in itself where each function symbol $\f$ other than $\f[exists]$ is interpreted by itself
  --- more precisely, the term $\f(x_1, \dots, x_n)$ ---
  and $\f[exists]$ is interpreted by the term $\und{1}$.
  This  interpretation of the axioms evidently holds in the initial model of $\thy$
  if, and only if,
  $\f[try](t, \und{0}) = \und{1}$ is provable for all closed terms $t$.

  Next, assume that the sentence $\fra{x} \exs{y} \varphi(x, y)$ is true and choose $m$.
  Then there exists an $n$ such that $\f_k(\und{m}, \und{n}) = \und{1}$ is provable in $\thy$.
  The following sequence of equations is then provable in $\thy$:
  \begin{align*}
    \f[try](\und{m}, \und{0})
    &= \f_k(\und{m}, \und{0}) \lor \f[try](\und{m}, \und{1}) \\
    &= \f_k(\und{m}, \und{0}) \lor \f_k(\und{m}, \und{1}) \lor \f[try](\und{m}, \und{2})  \\
    &\;\,\vdots \\
    &= \f_k(\und{m}, \und{0}) \lor \cdots \lor \f_k(\und{m}, \und{n}) \lor \f[try](\und{m}, \und{n + 1}) \\
    &= \und{1}
  \end{align*}
  (with the last holding as $\f_k(\und{m}, \und{n}) = \und{1}$ is provable).
  More generally, as $\f_k(\und{m}, \und{n'}) = \und{1}$ whenever $n' \geq n$,
    a similar argument shows that
    $\f[try](\und{m}, \und{n}) = \und{1}$ is provable for any $m$ and $n$.
  It follows that all closed terms of $\thy$ are provably equal to numerals.
  Therefore, as we have also shown that
    $\f[try](\und{m}, \und{0}) = \und{1}$ is provable for all $m$,
    the interpretation of the axioms holds in the initial model.

  Conversely, assume that $\f[try](t, \und{0}) = \und{1}$ is provable in $\thy$ for all closed terms $t$.
  Then, in particular, $\f[try](\und{m}, \und{0}) = \und{1}$ is provable in $\thy$ for all $m$.
  We analyse proofs in $\thy$ via the first-order term rewriting system obtained by orienting all the axioms of $\thy$ from left to right.
  As the reduction rules are left-linear and there are no overlaps between them, we have a Church-Rosser system~\cite{terese03term}.
  Hence, for any two terms $t$ and $u$, the equation $t = u$ is provable if, and only if, $t$ and $u$ reduce to a common term.

  In particular, by our assumption, for every  $m$, there is some number of steps $s$ such that
    $\f[try](\und{m}, \und{0}) \to^s \und{1}$.
  We first observe that $t \lor u \to^* \und{1}$ if, and only if, $t \to^* \und{1}$ or $u \to^* \und{1}$.
  Using this observation, it follows by induction on $s$ that
    there exists an $n$ such that $\f_k(\und{m}, \und{n}) \to^* \und{1}$.
  Therefore the sentence $\fra{x} \exs{y} \varphi(x, y)$ is true, concluding the proof.
\end{proof}

\begin{proof}[Proof of Theorem~\ref{thm:pi2}]
  Given a simple signature and theory $\thy$,
    following the proof of Proposition~\ref{prop:free}
    one recursively constructs a corresponding finitary equational theory $\thy_f$
    with finite signature and finitely many axioms whose models are in 1-1 correspondence with those of the given theory.
  Handlers $\ent H \T \handtype{F \type{0}\,}$ then correspond to interpretations of $\thy_f$ in itself,
    and are correct if, and only if, the interpretation of $\thy_f$ in itself holds in the model of $\thy_f$
    corresponding to the model $F\itp{\type{0}} = \itp{F\type{0}}$ of $\thy$.
  As $F\itp{\type{0}}$ is the initial model of $\thy$ (being its free model on the empty set),
    that corresponding model of $\thy$ is its initial model.

  Conversely, given a finitary equational theory $\thy$ with finite signature and finitely many axioms,
    one immediately constructs a simple signature and theory $\thy_s$
    such that the above construction yields back the original finitary equational theory $\thy$
    (up to a bijection of its function symbols).
  Further, given an interpretation of $\thy$ in itself one immediately constructs
  a handler $\ent H \T \handtype{F \type{0}\,}$ to which the interpretation of $\thy$ in itself  corresponds.

  The result then follows by combining these facts with Lemma~\ref{lem:pi2}.
\end{proof}

\begin{lem}
  \label{lem:sigma1}
  Given a finitary equational theory with finite signature and finitely many axioms,
  and an interpretation of this theory in itself,
  it is a $\Sigma_1$-complete problem to decide  whether the interpretation is provable in the theory.
\end{lem}

\begin{proof}
  \newcommand{\und}[1]{\underline{#1}}
  The proof is very similar to that of Lemma~\ref{lem:pi2}.
  The problem is clearly in $\Sigma_1$.
  Conversely, consider a $\Sigma_1$ sentence of Peano Arithmetic.
  This can be taken, without loss of generality,
    to be of the form $\exs{y} \varphi(y)$,
    where $\varphi(y)$ defines a primitive recursive predicate $R$.

  One then defines a finitary equational theory $\thy$ with finite signature and finitely many axioms,
    and an interpretation of it in itself,
    such that that  $ \exs{y} \varphi(y)$ is true
    if, and only if,
    the interpretation of the axioms holds in all free models of the theory over finite sets.
  The theory $\thy$ is much like before, except that $\f[try]$ is unary and $\f[exists]$ is a constant,
  the corresponding axioms are
  \begin{align*}
    \f[try](y) &= \f_k(y) \lor \f[try](\f[succ](y)) \\
    \f[exists]() &= \f[try](\und{0})
  \end{align*}
  and one considers the interpretation where each function symbol other than $\f[exists]$ is interpreted by itself
  and $\f[exists]$ is interpreted by $\und{1}$.
  This interpretation of the axioms evidently holds in all free models of $\thy$ over finite sets
  if, and only if,
  $\f[try](\und{0}) = \und{1}$ is provable.

  The rest of the proof then proceeds entirely analogously to, but a little more simply than, that of Lemma~\ref{lem:pi2}.
\end{proof}

\begin{proof}[Proof of Theorem~\ref{thm:sigma1}]
The proof proceeds analogously to that of Theorem~\ref{thm:pi2} except that one considers uniformly simple handlers
  $\{\ent H_{\C} \T \handtype{\C}\}_{\C}$.

In particular there is, as before, a finitary equational theory $\thy_f$ corresponding to any given theory $\thy$ over a given simple signature.
As before,  a given uniformly simple handler
  $\{\ent H_{\C} \T \handtype{\C}\}_{\C}$
then corresponds to an interpretation of $\thy_f$ in itself.
The handler is correct if that interpretation holds in all models of $\thy_f$ corresponding to those of $\thy$ of the form $\itp{\C}$.
Since the $\itp{\C}$ include all free models of $\thy$ of the form $F\type{n}$,
  this last is equivalent to the interpretation holding in all free models of $\thy_f$ over finite sets,
  and so to its being provable in $\thy_f$.
\end{proof}

\begin{proof}[Proof of Theorem~\ref{thm:decidable}]
  Following the proof of Theorem~\ref{thm:sigma1}, we see  that,
    given a theory $\thy$ over a given simple signature,
    determining the correctness of a given uniformly simple handler
    $\{ \ent H_{\C} \T \handtype{\C}\}_{\C}$
    amounts to determining whether the corresponding interpretation of $\thy_f$ in itself is provable in $\thy_f$.
  However that is decidable since $\thy_f$ is.
\end{proof}

\end{document}

%% file: definitions.tex
\def\emptyarg{}
\newcommand{\itp}[2][]{\llbracket #2 \rrbracket_{#1}}

\newcommand{\C}{\underline{C}}

\newcommand{\type}[1]{\mathbf{#1}}
\newcommand{\unit}{\type{1}}
\newcommand{\void}{\type{0}}
\newcommand{\nat}{\type{nat}}
\newcommand{\bool}{\type{bool}}
\newcommand{\exc}{\type{exc}}
\newcommand{\loc}{\type{loc}}
\newcommand{\chr}{\type{chr}}
\newcommand{\prodtype}[1]{\textstyle{\prod}_{#1}}
\newcommand{\sumtype}[1]{\textstyle{\sum}_{#1}}
\newcommand{\optype}[2]{#1 \rightarrowtriangle #2}
\newcommand{\conttype}[2]{#1 \to #2}
\newcommand{\handtype}[1]{#1\> \type{handler}}

\newcommand{\kord}[1]{\mathsf{#1}}
\newcommand{\kop}[1]{\>\mathsf{#1}\>}
\newcommand{\kpre}[1]{\mathsf{#1}\>}
\newcommand{\abs}[1]{{#1 .\>}}
\newcommand{\f}[1][f]{\mathsf{#1}}

\newcommand{\tru}{\kord{true}}
\newcommand{\fls}{\kord{false}}

\newcommand{\one}{\pair{}}
\newcommand{\pair}[1]{\langle #1 \rangle}

\newcommand{\thunk}{\kpre{thunk}}

\newcommand{\maps}{\mapsto}

\newcommand{\bind}[2]{#1 \kop{to} \abs{#2}}
\newcommand{\seq}{;\>}
\newcommand{\letin}[2]{\kpre{let} #1 \kop{be} #2 \kop{in}}
\newcommand{\lam}[1]{\mathord{\lambda \abs{#1}}}
\newcommand{\ret}{\kpre{return}}
\newcommand{\gen}[1]{\underline{\opsym{#1}}}
\newcommand{\opsym}[1]{\mathsf{#1}}
\newcommand{\opterm}[2]{\opsym{#1}_{#2}}
\makeatletter
\newcommand{\newop}[1]{\@ifstar{\opsym{#1}}{\opterm{#1}}}
\makeatother
\newcommand{\op}{\newop{op}}
\newcommand{\handleop}{\newop{handle}}
\newcommand{\raiseop}{\newop{raise}}
\newcommand{\chooseop}{\newop{choose}}
\newcommand{\getop}{\newop{get}}
\newcommand{\setop}{\newop{set}}
\newcommand{\delayop}{\newop{delay}}
\newcommand{\readop}{\newop{read}}
\newcommand{\writeop}{\newop{write}}
\newcommand{\botop}{\newop{div}}
\newcommand{\ifThenelse}[3]{\kpre{if} #1 \kop{then} #2 \kop{else} #3}
\newcommand{\force}{\kpre{force}}
\newcommand{\match}[1]{\kpre{match} #1 \kop{with}}
\newcommand{\matchcase}[3]{\match{#1} \{ #2 \}_{#3}}

\newcommand{\matchpair}[3]{\match{#1} \pair{#2} \maps #3}
\newcommand{\prj}[1]{\mathord{\mathsf{prj}}_{#1}\,}
\newcommand{\handle}[2]{#1 \kop{handled} \kpre{with} #2}
\newcommand{\handleto}[3]{\bind{\handle{#1}{#2}}{#3}}
\newcommand{\rec}[1]{\kpre{rec} \abs{#1}}
\renewcommand{\H}[1]{H_{\text{#1}}}

\newcommand{\thy}{\mathcal{T}}
\newcommand{\ctx}{\Gamma}
\newcommand{\tctx}{Z}
\newcommand{\kctx}{K}
\newcommand{\T}{\mathop{:}}
\newcommand{\stoup}{\mid}
\newcommand{\ent}{\vdash}
\newcommand{\cond}[1]{\quad (#1)}
\newcommand{\deduct}[3][]{%
  \def\arg{#1}%
  \inferrule{#2}{#3}%
  \ifx\arg\emptyarg\else\cond{#1}\fi%
}
\newcommand{\defeq}{\stackrel{\mathrm{def}}{=}}
\newcommand{\bnfis}{\mathrel{{:}{:}\!=}}
\newcommand{\bnfor}{\mathrel{\mid}}
\newcommand{\from}{\colon}

\newcommand{\Cat}[1]{\mathord{\mathbf{#1}}}
\newcommand{\Set}{\Cat{Set}}
\newcommand{\wCpo}{\Cat{\omega Cpo}}
\newcommand{\Mod}[1][]{\Cat{Mod}_{#1}}
\newcommand{\Map}[2][]{\mathop{\mathrm{#2}_{#1}}\nolimits}
\renewcommand{\Pr}[1]{\Map[#1]{pr}}
\newcommand{\In}[1]{\Map[#1]{in}}
\newcommand{\Ev}{\mathord{\mathrm{ev}}}
\newcommand{\Ds}{\mathord{\mathrm{dist}}}
\newcommand{\Tr}[1]{\Map{tr}(#1)}
\newcommand{\Id}[1][]{\mathord{\mathrm{id}_{#1}}}

\newcommand{\St}[1]{\mathord{\mathrm{st}}}
\newcommand{\model}{\mathcal{M}}
\newcommand{\homo}{h}
\newcommand{\lol}{\multimap}
\newcommand{\fra}[1]{\forall \abs{#1}}
\newcommand{\exs}[1]{\exists \abs{#1}}